\documentclass[11pt]{article}
\usepackage{fullpage}
\usepackage{amsfonts}
\usepackage{amssymb,amsmath,amsthm}
\usepackage{graphicx}
\usepackage[ruled, vlined,linesnumbered]{algorithm2e}
\usepackage{algpseudocode}
\usepackage{mathpazo,bbm}

\AtBeginDocument{
\DeclareSymbolFont{AMSb}{U}{msb}{m}{n}
    \DeclareSymbolFontAlphabet{\mathbb}{AMSb}}

\newtheorem{theorem}{Theorem}[section]

\newtheorem{lemma}[theorem]{Lemma}

\newtheorem{defn}[theorem]{Definition}

\newtheorem{claim}[theorem]{Claim}
\newtheorem{proposition}[theorem]{Proposition}
\newtheorem{corollary}[theorem]{Corollary}
\renewcommand{\epsilon}{\varepsilon}

\renewcommand{\dim}{\mathsf{dim}}

\newcommand{\F}{{\mathbb F}}
\newcommand{\E}{{\mathbf E}}
\renewcommand{\Pr}{{\mathbf{Pr}}}

\renewcommand{\le}{\leqslant}
\renewcommand{\ge}{\geqslant}
\renewcommand{\leq}{\leqslant}
\renewcommand{\geq}{\geqslant}

\usepackage{color}

\newcommand{\ur}{\omega}

\newcommand{\ignore}[1]{}

\newcommand{\monomialeq}[3]{\mathcal{M}_{#1}^{=#2}(#3)}
\newcommand{\monomialle}[3]{\mathcal{M}_{#1}^{\leq #2}(#3)}

\newcommand{\generatingmatrix}[1]{M^{(#1)}}
\newcommand{\rounditerator}{t}

\newcommand{\groupheight}[1]{h_{#1}}

\newcommand{\func}[2]{g_{#1,#2}}
\newcommand{\cone}{c_{1}}
\newcommand{\ctwo}{c_{2}}
\newcommand{\tlength}{t'}
\newcommand{\ttwo}{t''}
\newcommand{\hone}{k}
\newcommand{\htwo}{k'}
\newcommand{\done}{d'}

\newcommand{\dtwon}{d'-t'}
\newcommand{\csum}{\xi}
\newcommand{\biasj}{\mathrm{bias}_j}

\newcommand{\negbiasj}{\mathrm{bias}_{-j}}
\newcommand{\ind}[1]{^{(#1)}}
\newcommand{\floor}[1]{\lfloor #1\rfloor}
\newcommand{\constnvsn}{\gamma}
\newcommand{\Pdm}{\mathcal{P}_p(d,n)}

\newcommand{\monomial}[3]{\mathcal{M}_{#1}(#2,#3)}

\newcommand{\sampledim}{n}
\newcommand{\hypothesisdim}{m}
\newcommand{\iterator}{m}

\begin{document}
\title{On the Bias of Reed-Muller Codes over Odd Prime Fields}

\author{Paul Beame\thanks{Research supported in part by NSF grant CCF-1524246} \\ University of Washington \\ beame@cs.washington.edu \and Shayan Oveis Gharan\thanks{Research supported in part by NSF grant CCF-1552097 and ONR-YI grant N00014-17-1-2429}\\ University of Washington \\ shayan@cs.washington.edu \and Xin Yang$^*$ \\University of Washington\\yx1992@cs.washington.edu}

\date{\today}
\maketitle
\abstract 
We study the bias of random bounded-degree polynomials over odd prime fields and show
that, with probability exponentially close to 1, such polynomials have exponentially
small bias.
This also yields an exponential tail bound on the weight
distribution of Reed-Muller codes over odd prime fields.
These results generalize bounds of Ben-Eliezer, Hod, and Lovett who proved similar
results over $\mathbb{F}_2$.
A key to our bounds is the proof of a new precise extremal property for the 
rank of sub-matrices of the generator matrices of Reed-Muller codes over odd prime
fields.  This extremal property is a substantial extension of an 
extremal property shown by Keevash and Sudakov for the case of $\mathbb{F}_2$. 

Our exponential tail bounds on the bias can be used to derive exponential lower
bounds on the time for space-bounded learning of bounded-degree polynomials from
their evaluations over odd prime fields.
  
\newpage
\newcommand{\bias}{\mathrm{bias}}

\section{Introduction}

Reed-Muller codes are among the oldest error correcting codes,
first introduced by Muller \cite{muller1954application} and 
Reed \cite{reed1953class} in the 1950s.  
These codes were initially defined in terms of bounded-degree
multivariate polynomials
over $\mathbb{F}_2$ but the same definition can be applied over any 
finite field.
To be more precise, the $(d,\sampledim)$ Reed-Muller code over finite field 
$\mathbb{F}$, denoted $RM_\mathbb{F}(d,\sampledim)$, takes the message as
the coefficients of some $\sampledim$-variate polynomial of degree at
most $d$ over $\mathbb{F}$,
and the encoding is simply the evaluation of that polynomial over all possible
inputs chosen from $\mathbb{F}^\sampledim$.

A function $f:\mathbb{F}^\sampledim\rightarrow \mathbb{F}$ is
\emph{balanced} if elements of $\mathbb{F}$ occurs an equal number of
times as an output of $f$.
The bias of a function $f$ with co-domain $\mathbb{F}$ is a measure of 
the fractional deviation of $f$ from being balanced.
Since each codeword in a Reed-Muller code is the evaluation of a
(polynomial) function over all elements of its domain, the definition
of bias directly applies to the codewords of a Reed-Muller code.

Some elements of a Reed-Muller code are very far from balanced (for 
example the 0 polynomial yields the all-0 codeword, and the codeword
for the polynomial $1+x_1 x_2$ has value 1 much more frequently than 
average) but since, as we might expect, randomly-chosen polynomials behave somewhat like randomly-chosen functions,
most codewords are close to being balanced.    
We quantify that statement and show that for all prime fields, only an 
exponentially small fraction of Reed-Muller codewords (equivalently, an
exponentially small fraction of polynomials of bounded degree) have as
much an exponentially small deviation from perfect balance.  That is, 
at most an exponentially small fraction of polynomials have more than
an exponentially small bias.  Such a result is already known for the
case of $\mathbb{F}_2$~\cite{DBLP:journals/cc/Ben-EliezerHL12} so
we will only need to prove the statement for odd prime fields.

We now define bias formally and discuss its applications.
In the case that $f:\mathbb{F}_2^\sampledim\rightarrow \mathbb{F}_2$,
the \emph{bias} of $f$,
\[
\bias(f):=\frac 1 {2^\sampledim}\sum_{x\in \mathbb{F}_2^\sampledim} (-1)^{f(x)}=\Pr_{x\in_R \mathbb{F}_2^\sampledim}[f(x)=0]-\Pr_{x\in_R \mathbb{F}_2^\sampledim}[f(x)=1].
\]
More generally, for $p$ a prime, $\omega=e^{2\pi i/p}$,
and $j\in \mathbb{F}_p^*$, we define 
the \emph{$j$-th order bias} of 
$f:\mathbb{F}_p^\sampledim\rightarrow \mathbb{F}_p$
as
\[
\biasj(f):=\frac 1 {p^\sampledim}\sum_{x\in \mathbb{F}_p^\sampledim} \omega^{j\cdot f(x)}.
\]
Prior uses of bias over these larger co-domains often
focus only on the case of a single $j$~(e.g., \cite{DBLP:journals/tit/BhowmickL18,haramaty2010structure}) since they consider structural
implications of bias. 
However, the use of different values of $j$ is essential
for the  applications of bias to bounding the imbalance 
of functions and codewords since, for $p>3$, one can have
functions with 1st-order bias 0 that are very far from
balanced.  
It turns out that it is necessary and sufficient to bound 
$|\biasj(f)|$ for all $j\in \mathbb{F}^*_p$ 
(or, equivalently, all integers $j$ with
$1\le j\le (p-1)/2$ since $|\biasj(f)|=|\negbiasj(f)|$)
in order to bound the imbalance: A standard exponential 
summation argument (e.g., Proposition 2.1
in \cite{bogy:learning-coltfull-tr,bogy:learning-colt}), shows that for
every $b\in \mathbb{F}_p$, 
$$\left|\Pr_{x\in_R \mathbb{F}_p^\sampledim}[f(x)=b]-\frac{1}{p}\right|\le \max_{j\in \mathbb{F}^*_p}|\biasj(f)|.$$

For Reed-Muller codes, the bias of a codeword
exactly determines its fraction (number of non-zero entries, which
is called the \emph{weight} of the codeword.
(In the case of $\mathbb{F}_2$ the bias is determined
by the weight but that is not true for $\mathbb{F}_p$ for odd prime $p$.)
The distribution of weights of codewords in Reed-Muller codes over
$\mathbb{F}_2$ plays a critical role in many applications in coding
theory and in many other applications in theoretical computer 
science.  
As a consequence, the weight distribution of Reed-Muller codes
over $\mathbb{F}_2$ has been the subject of considerable study.
For degrees $d=1$ and $d=2$, the exact weight distribution (and hence
the distribution of the bias) for
$RM_{\mathbb{F}_2}(2,\sampledim)$ has been known for roughly
50 years~\cite{DBLP:journals/tit/SloaneB70,mceliece1967linear}.
For other degrees, precise bounds are only known for weights up to 2.5 
times the minimum distance of such
codes~\cite{kasamiT70,kasamiTA76} but this is very far from the balanced
regime.

For general constant degrees, Kaufman, Lovett and
Porat~\cite{DBLP:journals/tit/KaufmanLP12} give a somewhat tight bound
on the weight distribution for Reed-Muller codes over $\mathbb{F}_2$, and Abbe, Shpilka, and 
Wigderson~\cite{abbe2015reed} generalize the result to linear degrees.
These results yield tail bounds for the number of codewords with bias approaching 0 and, 
using the cases for arbitrarily small constant bias, imply good bounds for list-decoding
algorithms~\cite{gopalan2008list,DBLP:journals/tit/KaufmanLP12}.

Ben-Eliezer, Hod, and Lovett~\cite{DBLP:journals/cc/Ben-EliezerHL12} proved sharper 
bounds showing that the fraction of codewords with more 
than exponentially small bias (of the
form $2^{-c_1 n/d)}$ for constant $c_1>0$) is at most $2^{-c_2\hypothesisdim}=|RM_{\mathbb{F}_2}(d,\sampledim)|^{-c_2}$ for constant $c_2>0$ where 
$\hypothesisdim=\log_2|RM_{\mathbb{F}_2}(d,\sampledim)|$
is the dimension of the code. 
(For $d< \sampledim/2$ they also showed that this fraction of codewords is tight 
by exhibiting a set of codewords in 
$RM_{\mathbb{F}_2}(d,\sampledim)$ of size 
$|RM_{\mathbb{F}_2}(d,\sampledim)|^{c_3}$ for $c_3>0$ that
has such a bias.)
This bound was used by \cite{bogy:learning-coltfull-tr,bogy:learning-colt,grt:extractor-learn-stoc} to show
that learning bounded degree polynomials over 
$\mathbb{F}_2$ from their evaluations with success 
probability $2^{-o(\sampledim)}$ requires space
$\Omega(\sampledim \hypothesisdim/d)$ or time
$2^{\Omega(\sampledim/d)}$.

\paragraph{Our Results}
 
We generalize the results of Ben-Eliezer, Hod, and Lovett~\cite{DBLP:journals/cc/Ben-EliezerHL12} to show 
that only an exponentially small fraction of polynomials over prime fields can have non-negligible bias.
Formally speaking,
let $\Pdm$ denote the set of polynomials of degree at most $d$
in $\sampledim$ variables over $\mathbb{F}_p$,
and let $\monomial{p}{d}{\sampledim}$ denote the set of monic monomials of degree at most $d$ in $\sampledim$ variables. 
(The Reed-Muller code $RM_{\mathbb{F}_p}(d,\sampledim)$ 
has dimension $|\monomial{p}{d}{\sampledim}|$ and
satisfies $|RM_{\mathbb{F}_p}(d,\sampledim)|=|\Pdm|$.)

Our main result is the following theorem:
\begin{theorem}\label{lem:main-bias}
For any $0<\delta<1/2$ there are constants $c_1,c_2>0$ depending on 
$\delta$ such that for any odd prime $p$,
for all integers $d\leq \delta \sampledim$ and all
 $j\in \F_p^*$, we have
\[
\Pr_{f\in_R \Pdm}[|\biasj(f)|>p^{-c_1 \sampledim/d}]\leq p^{-c_2|\monomial{p}{d}{\sampledim}|}.
\]
\end{theorem}

Using this theorem together with the methods of our companion 
paper~\cite{bogy:learning-coltfull-tr,bogy:learning-colt} or
of~\cite{grt:extractor-learn-stoc}, we obtain that any
algorithm that learns polynomials over $\mathbb{F}_p$ of degree at
most $d$ with probability at least
$p^{-O(\sampledim)}$ from their evaluations on random inputs
 either requires time $p^{\Omega(\sampledim/d)}$ or space 
$\Omega(\sampledim\cdot |\monomial{p}{d}{\sampledim}|/d\cdot \log p)$.   
For the details, see~\cite{bogy:learning-coltfull-tr}.

The following corollary of Theorem~\ref{lem:main-bias} is also immediate:

\begin{corollary} \label{cor:weight} For any $0<\delta<1/2$ there are constants $c_1, c_2>0$ 
such that for any odd prime $p$ and integers $d$, $\sampledim$ with
$d\le \delta\sampledim$, the number of 
codewords of $RM_{\mathbb{F}_p}(d,\sampledim)$ of weight
at most $1-1/p-p^{-c_1 \sampledim/d}$ is at most
$|RM_{\mathbb{F}_p}(d,\sampledim)|^{1-c_2}$.
\end{corollary}

There is a limit to the amount that Theorem~\ref{lem:main-bias} 
can be improved, as shown by the following proposition:

\begin{proposition}
\label{prop:lower-bias}
For any $0< \delta <1/2$ there are constants $c'<1$ and $c''>0$ 
depending on $\delta$ such that for all integers 
$d\le \delta \sampledim$
and all $j\in \mathbb{F}_p^*$, we have
$$\Pr_{f\in_R\Pdm}[|\biasj(f)|>p^{-c''n/d}]\ge p^{-c'|\monomial{p}{d}{\sampledim}|}.$$
\end{proposition}

As part of our proof of Theorem~\ref{lem:main-bias},
we must prove the following tight bound on the rank of the evaluations of monomials of 
degree at most $d$ on sets of points.
Alternatively this can be seen as the extremal dimension of the span of
truncated Reed-Muller codes at sizes that are powers of the field size.

\begin{lemma}\label{prop:evaluation-rank}
Let $S$ be a subset of $\F_p^\sampledim$ such that $|S|=p^r$.
Then the dimension of the subspace spanned by $\{(q(x))_{q\in \monomial{p}{d}{\sampledim}}:x\in S\}$ is at least $|\monomial{p}{d}{r}|$.
\end{lemma}

Though this is all that we require to prove Theorem~\ref{lem:main-bias}, we prove 
it as a special case of a more general theorem that gives an exact extremal
characterization of the dimension of the span of truncated Reed-Muller codes of
all sizes.  This generalizes a characterization for the case
of $\mathbb{F}_2$ proved by Keevash and Sudakov~\cite{keevash2005set}.

\begin{theorem}
\label{thm:main-extremal}
Let $1\le m\le p^r$ and let $\sampledim\ge r$. 
For $S\subseteq \F_p^\sampledim$ with $|S|=m$,
$$\dim\langle\{(q(x))_{q\in \monomial{p}{d}{\sampledim}}\ :\  x\in S\}\rangle \ge
\dim\langle\{(q(x))_{q\in\monomial{p}{d}{r}}\ :\ x\in T\}\rangle,$$
where $T$ consists of the $m$ lexicographically minimal vectors in $\F_p^r$.  
(This is equality when $S$ is also lexicographically minimal.)
\end{theorem}

Thus, the extremal value of the dimension is a function $g_d(m)$ that is 
independent of $\sampledim$.  As part of the proof of 
Theorem~\ref{thm:main-extremal}, we characterize a variety of properties 
of $g_d(m)$.

\paragraph{Proof Overview}

Our basic approach is a generalization of the high level outline of 
\cite{DBLP:journals/cc/Ben-EliezerHL12} to odd prime fields, though parts of the
argument are substantially more complex:

We begin by using a moment method, showing that
that $\E_{f\in_R \Pdm}[|\biasj(f)|^t]$ is bounded for suitable $t$.
Because we are dealing with odd prime fields rather than $\F_2$ we restrict
ourselves to the case that $t$ is even.
For bounding these high moments,
we reduce the problem to lower bounding the rank of certain random matrices (Lemma~\ref{lem:rank}).
This is the place where we can apply Lemma~\ref{prop:evaluation-rank} to prove the bound.

\begin{sloppypar}
For the case of $\F_2$ handled in~\cite{DBLP:journals/cc/Ben-EliezerHL12}, 
a similar property to
Lemma~\ref{lem:rank} (Lemma 4 in \cite{DBLP:journals/cc/Ben-EliezerHL12}), 
which follows from an extremal characterization of $\F_2$
polynomial evaluations by Keevash and Sudakov~\cite{keevash2005set}, was independently
shown to follow more simply via an 
algorithmic construction that avoids 
consideration of any subset size that is not a power of 2.
Unfortunately, this simpler algorithmic construction seems to 
break down completely for the case of odd prime fields.
\end{sloppypar}

We instead provide the full extremal characterization for all set sizes, analogous
to the Keevash and Sudakov characterization for $\F_2$.  This is the major source
of technical difficulty in our paper.  Like Keevash and Sudakov, we show that the
proof of our extremal characterization is equivalent to proving the sub-additivity
of a certain arithmetic function.  However, proving this sub-additivity property
is an order of magnitude more involved since it involves sub-additivity over $p$
terms for arbitrary $p$ rather than just over the two terms required for the case
of $\F_2$. 

\paragraph{Discussion and Related Work}

Prior to our work, the main approach to analyzing 
the bias of polynomials over arbitrary prime fields
has been to take a structural point of view.  
The general idea is to show that 
polynomials of large bias must have this bias because
of some structural property.
For polynomials of degree $d=2$, a complete structural
characterization has
been known for more than a century (\cite{dickson:book}).
Green and Tao~\cite{green2007distribution} initiated the 
modern study of the relationship between the bias and the 
structure of polynomials over finite fields.
Kaufman, Lovett, and Porat~\cite{DBLP:journals/tit/KaufmanLP12} used this approach to obtain their bounds on bias over $\mathbb{F}_2$.
Over general prime fields,
Haramaty and Shpilka~\cite{haramaty2010structure} gave 
sharper structural properties for polynomials of degrees $d=3,4$.
In papers~\cite{DBLP:journals/tit/BhowmickL18} for constant degree and \cite{DBLP:journals/corr/0001L15} for
large degree, Bhowmick and Lovett generalized the result 
of \cite{DBLP:journals/tit/KaufmanLP12} to show that if a 
degree $d$ polynomial $f$ has
large bias, then $f$ can be expressed as a function of a
constant number of polynomials of degree at most $d-1$.
These bounds are sufficient to analyze the list-decoding 
properties of Reed-Muller codes.
However, all of these structural results, except for
the characterization of degree 2 polynomials, are too 
weak to obtain the bounds on sub-constant bias that we
derive.   Indeed, none is sufficient even to derive
Corollary~\ref{cor:weight}.

An open problem that remains from our work, as well as
that of Ben-Eliezer, Hod, and
Lovett~\cite{DBLP:journals/cc/Ben-EliezerHL12} is whether
the amount of the bias can be improved still further by 
removing the $1/d$ factor from the exponent in the bias
in the statement of Theorem~\ref{lem:main-bias} for some
range of values of $d$ growing with $\sampledim$.   
Though Proposition~\ref{prop:lower-bias} (and its 
analogue in~\cite{DBLP:journals/cc/Ben-EliezerHL12})
show that a
large number of polynomials have bias 
$p^{-O(\sampledim/d)}$, we would need to extend them to
say that for \emph{all} $c'>0$ there is a $c''>0$ such 
that the conclusion of the proposition holds in order
to rule out improving the bias in Theorem~\ref{lem:main-bias}.

\paragraph{Organization}
The proof of Theorem \ref{lem:main-bias},
except for the proof of Lemma \ref{prop:evaluation-rank},
is in Section \ref{sec:random-bias}.   
Section~\ref{sec:random-bias} also contains the proof of Proposition~\ref{prop:lower-bias}.
In Section \ref{sec:reed-muller} we reduce the proof of 
Lemma~\ref{prop:evaluation-rank}, and that of the general
extremal rank property of 
Theorem~\ref{thm:main-extremal}, to
proving the sub-additivity of the
arithmetic function $g_d$.
In Section \ref{sec:g-d-property} we introduce some properties of $g_d$,
and finally in Section~\ref{sec:sub-add} we prove the sub-additivity of $g_d$.

\section{The bias of random polynomials over odd prime fields}
\label{sec:random-bias}

In this section we prove Theorem \ref{lem:main-bias}.
To provide tail bounds on the bias, we first characterize its high moments,
focusing on even moments to ensure that they are real-valued.

\begin{lemma}\label{lem:expectation-dual}
Let $p$ be an odd prime and $d\le \sampledim$.
For $t\in \mathbb{N}$,
let $x\ind{1},\cdots,x\ind{t}$ and $y\ind{1},\cdots,y\ind{t}$ be chosen
uniformly at random from $\mathbb{F}_p^\sampledim$.
Then 
\[
\E_{f\in_R \Pdm}[\;|\biasj(f)|^{2t}\;]=\Pr_{x\ind{1},\cdots,x\ind{t},y\ind{1},\cdots,y\ind{t}} [\;\forall q\in \monomial{p}{d}{\sampledim},\ \ \sum_{k=1}^t q(x\ind{k})=\sum_{k=1}^t q(y\ind{k})\;].
\]
\end{lemma}

\begin{proof}
Note that $\overline{\biasj(f)}=\overline{\E_{x}[\ur^{j\cdot f(x)}]}=\E_{x}[\overline{\ur^{j\cdot f(x)}}]=\E_{x}[\ur^{-j\cdot f(x)}]=\negbiasj(f)$,
therefore $|\biasj(f)|^2=\biasj(f)\cdot \negbiasj(f)$.
So we have 
\begin{align*}
\E_{f\in_R \Pdm}[|\biasj(f)|^{2t}]&=\E_{f\in_R \Pdm}[\biasj(f)^t\cdot \negbiasj(f)^t]\\
&=\E_{f\in_R \Pdm}[\prod_{k=1}^t\E_{x\ind{k}}[\ur^{j\cdot f(x\ind{k})}]\cdot \prod_{k=1}^t\E_{y\ind{k}}[\ur^{-j\cdot f(y\ind{k})}]]\\
&=\E_{f\in_R \Pdm}[\E_{x\ind{1},\cdots,x\ind{t},y\ind{1},\cdots,y\ind{t}}[\ur^{j\cdot(\sum_{k=1}^tf(x\ind{k})-\sum_{k=1}^tf(y\ind{k}))}]]\\
&=\E_{x\ind{1},\cdots,x\ind{t},y\ind{1},\cdots,y\ind{t}}[\E_{f\in_R \Pdm}[\ur^{j\cdot(\sum_{k=1}^tf(x\ind{k})-\sum_{k=1}^tf(y\ind{k}))}]]\\
\end{align*}
For each $q\in\monomial{p}{d}{\sampledim}$ let $f_q\in \mathbb{F}_p$ denote the
coefficient of $q$ in $f$.
We identify $f$ with its vector of coefficients $(f_q)_{q\in \monomial{p}{d}{\sampledim}}$ and choose $f$ uniformly by choosing the $f_q$ uniformly.
Therefore
\begin{align*}
\E_{f\in_R \Pdm}[|\biasj(f)|^{2t}]&=\E_{x\ind{1},\cdots,x\ind{t},y\ind{1},\cdots,y\ind{t}}[\E_{f\in_R \Pdm}[\ur^{j\cdot(\sum_{q\in \monomial{p}{d}{\sampledim}}f_q\cdot (\sum_{k=1}^t q(x\ind{k})-\sum_{k=1}^t q(y\ind{k})))}]]\\
&=\E_{x\ind{1},\cdots,x\ind{t},y\ind{1},\cdots,y\ind{t}}[\prod_{q\in \monomial{p}{d}{\sampledim}}\E_{f_q\in_R \mathbb{F}_p}[\ur^{j\cdot f_q\cdot  (\sum_{k=1}^t q(x\ind{k})-\sum_{k=1}^t q(y\ind{k}))}]]\\
&=\E_{x\ind{1},\cdots,x\ind{t},y\ind{1},\cdots,y\ind{t}}[\mathbf{1}_{(\forall q\in \monomial{p}{d}{\sampledim},\ \sum_{k=1}^t q(x\ind{k})-\sum_{k=1}^t q(y\ind{k})=0)}]\\
&=\Pr_{x\ind{1},\cdots,x\ind{t},y\ind{1},\cdots,y\ind{t}}[\forall q\in \monomial{p}{d}{\sampledim},\ \sum_{k=1}^t q(x\ind{k})=\sum_{k=1}^t q(y\ind{k})]
\end{align*}
where the second equality follows since
$\E_{a\in_R \mathbb{F}_p}[\omega^{j\cdot a\cdot b}]=0$ for all
$b\in \mathbb{F}^*_p$.
\end{proof}

Now let us look at the probability $$\Pr_{x\ind{1},\cdots,x\ind{t},y\ind{1},\cdots,y\ind{t}}[\forall q\in \monomial{p}{d}{\sampledim},\ \sum_{k=1}^t q(x\ind{k})=\sum_{k=1}^t q(y\ind{k})].$$
We view $y\ind{1},\cdots,y\ind{t}$ as arbitrary fixed values 
and we will upper bound this probability
following the analysis of a similar probability in
\cite{DBLP:journals/cc/Ben-EliezerHL12}.
That is, we will upper bound the probability that this holds by considering a
special subset $\mathcal{M}'\subseteq \monomial{p}{d}{\sampledim}$ that allows us
to derive a linear system whose rank will bound the probability that the
constraints indexed by $\mathcal{M}'$ all hold.

We divide $[\sampledim]$ arbitrarily into two disjoint parts $L$ and $R$ with $|L|=\floor{\frac{\sampledim}{d}}$.
$\mathcal{M}'\subseteq \monomial{p}{d}{\sampledim}$
consists of all monomials of degree are most $d$ that
have degree 1 on $L$ and degree at most $d-1$ on $R$.

We use the following properties of the $|\monomial{p}{d}{\sampledim}|$, whose proof we
defer to later, to show that $\mathcal{M}'$ contains a significant
fraction of all monomials in $\monomial{p}{d}{\sampledim}$.

\begin{proposition} \label{prop:d-vs-d-1}
If $d\leq \delta \sampledim$ for some $0<\delta < 1$ then
\begin{itemize}
\item[(a)] 
there exists a constant $\gamma'=\gamma'(\delta)>0$ such that for
sufficiently large $\sampledim$, 
if $\sampledim'\geq (1-\frac 1 d) \sampledim$ then
\[
|\monomial{p}{d}{\sampledim'}|\geq \gamma' |\monomial{p}{d}{\sampledim}|.
\]
\item[(b)] If $p\ge 3$ there exist constants $\rho_1,\rho_2>0$ such that for
sufficiently large $\sampledim$,
\[
\rho_1|\monomial{p}{d}{\sampledim}|\leq \frac{\sampledim}{d}\cdot |\monomial{p}{d-1}{\sampledim}|\leq \rho_2|\monomial{p}{d}{\sampledim}|.
\]
\end{itemize}
\end{proposition}

\begin{corollary}
\label{prop:mprime-bound}
Let $p\ge 3$. If $d\leq \delta \sampledim$ for some $0<\delta < 1$,
then there exists a constant $\gamma=\gamma(\delta)>0$  such that for sufficiently
large $\sampledim$,
\[
|\mathcal{M}'|=
\floor{\frac{\sampledim}{d}}\cdot |\monomial{p}{d-1}{\sampledim-\floor{\frac{\sampledim}{d}}}|
\ge
\gamma\cdot |\monomial{p}{d}{\sampledim}|.
\]
\end{corollary}
 
\begin{proof}
The equality follows immediately from the definition of $\mathcal{M}'$.
Let $\sampledim'=\sampledim-\floor{\frac{\sampledim}{d}}$.  Then
\begin{align*}
|\mathcal{M}'|&=\floor{\frac{\sampledim}{d}}\cdot |\monomial{p}{d-1}{\sampledim'}|\\
&\ge \frac{\sampledim'}{2d} |\monomial{p}{d-1}{\sampledim'}|\qquad\mbox{since $d\le \sampledim$}\\
&\ge \frac{\rho_1}{2} |\monomial{p}{d}{\sampledim'}|\qquad\mbox{by Proposition~\ref{prop:d-vs-d-1}(b)}\\
&\ge \frac{\rho_1\gamma'}{2} |\monomial{p}{d}{m}|\qquad\mbox{by Proposition~\ref{prop:d-vs-d-1}(a)}
\end{align*}
and setting $\gamma=\rho_1\gamma'/2$ yields the claim.
\end{proof}

Let $\mathcal{E}$ denote the event that 
$\sum_{k=1}^t q(x\ind{k})=\sum_{k=1}^t q(y\ind{k})$ for all $q\in \mathcal{M}'$.
To simply notation, since we think of $y\ind{1},\ldots,y\ind{k}$ as fixed, 
for each $q\in \mathcal{M}'$ define $b_q\in \mathbb{F}_p$ by
$b_q=\sum_{k=1}^t q(y\ind{k})$.
Since any $q\in \mathcal{M}'$ is of the form $q=x_i\cdot q'$ for some $i\in L$
and $q'$ a monomial of degree at most $d-1$ on $R$,
$\mathcal{E}$ requires that
\[
b_q=\sum_{k=1}^t q(x\ind{k})=\sum_{k=1}^t q'(x\ind{k}_R)\cdot x\ind{k}_i.
\]
where for $x\in \mathbb{F}_p^\sampledim$, we write $x_R$ for $x$ restricted to the 
coordinates in $R$.
We view these constraints as a system of linear equations over the set of
variables $x\ind{k}_i$ for $k\in [t]$ and $i\in L$ whose coefficients are given
by the values of $q'(x\ind{k}_R)$ for $x\ind{k}_R\in \mathbb{F}_p^R$
for all $k\in [t]$.
Observe that for different values of $i\in L$ we get separate and
independent subsystems of
equations with precisely the same coefficients but potentially
different constant terms $b_q$ since $q$ depends on both $i$ and $q'$.
Therefore the probability that $(x\ind{k}_i)_{i\in L,k\in [t]}$ is a solution
is the product of the probabilities for the individual choices of $i\in L$.

For each $\mathbf{x}_R=x\ind{1}_R,\ldots,x\ind{t}_R$, there is a 
$|\monomial{p}{d-1}{R}|\times t$ matrix $Q_{\mathbf{x}_R}$ for 
a system of linear equations on $(x\ind{1}_i,\ldots,x\ind{t}_i)$ for each
$i\in L$, having one constraint for each polynomial $q'$ of degree at most
$d-1$ on $R$.
Observe that $Q_{\mathbf{x}_R}(q',k)=q'(x\ind{k}_R)$.

In particular, it follows that 
\begin{equation}
\Pr_{x\ind{1},\cdots,x\ind{t},y\ind{1},\cdot,y\ind{t}}[\;\mathcal{E}\mid (x\ind{1}_R,\cdots,x\ind{t}_R)=\mathbf{x}_R]\le p^{-rank(Q_{\mathbf{x}_R})\cdot |L|}.
\label{eq:rank}
\end{equation}

We now see that for almost all choices of $\mathbf{x}_R$, if $t$ is at
least a constant factor larger than $|\monomial{p}{d-1}{|R|}|$ then the
rank of $Q_{\mathbf{x}_R}$ is large.  
This follows by replacing $\sampledim$ by $|R|$, $d$ by $d-1$, $q'$ by $q$ and
$\mathbf{x}$ by $\mathbf{x}_R$ in the following lemma.

\begin{lemma}\label{lem:rank}
For any $0<\delta\le 1/2$ there is a constant
$\constnvsn=\constnvsn(\delta)>0$ such that 
there exist constants $c>0$ and $\eta>1$ such that
for $d=\floor{\delta \sampledim}$ and $t\geq \eta |\monomial{p}{d}{\sampledim}|$, 
if $\mathbf{x}=x\ind{1},\cdots,x\ind{t}$ is chosen uniformly at random from
$(\F_p^\sampledim)^t$, then the matrix
$Q_{\mathbf{x}}\in \F_p^{\monomial{p}{d}{\sampledim}\times [t]}$ 
given by $Q_{\mathbf{x}}(q,k)=q(x\ind{k})$.
then
\[
\Pr_{\mathbf{x}}[\,rank(Q_\mathbf{x})\leq \constnvsn|\monomial{p}{d}{\sampledim}|\,]\leq p^{-c|\monomial{p}{d+1}{\sampledim}|}.
\]
\end{lemma}

We first show how to use Lemma~\ref{lem:rank} to prove
Theorem \ref{lem:main-bias}.

\begin{proof}[Proof of Theorem \ref{lem:main-bias}]
Let $0<\delta\le 1/2$, and set $\constnvsn>0$ and $\eta>1$ and $c>0$ as in
Lemma~\ref{lem:rank}.
Let $t=\lceil\eta \monomial{p}{d-1}{\sampledim}\rceil$.
We first bound the expected value of $|\biasj(f)|^{2t}$.  
By Lemma~\ref{lem:expectation-dual}
and the definition of event $\mathcal{E}$ we have
$$\E_{f\in_R \Pdm}[\;|\biasj(f)|^{2t}\;]=\Pr_{x\ind{1},\cdots,x\ind{t},y\ind{1},\cdot,y\ind{t}} [\;\mathcal{E} \;].$$
Let $\sampledim'=\lceil \sampledim(1-1/d)\rceil$ and $d'=d-1$.  
Now by definition,
\begin{align*}
&\Pr_{x\ind{1},\cdots,x\ind{t},y\ind{1},\cdot,y\ind{t}}[\;\mathcal{E}\;]\\
&\le\ \Pr_{\mathbf{x_R}}[rank(Q_{\mathbf{x}_R})\le \constnvsn |\monomial{p}{d}{\sampledim'}|]\\
&\qquad +\Pr_{x\ind{1},\cdots,x\ind{t},y\ind{1},\cdot,y\ind{t}}[\;\mathcal{E}\ :\  
rank(Q_{\mathbf{x}_R})\le 
\constnvsn |\monomial{p}{d'}{\sampledim'}|\mbox{ for }
\mathbf{x}_R=(x\ind{1}_R,\cdots,x\ind{t}_R)]\\
&\le\ \Pr_{\mathbf{x_R}}[rank(Q_{\mathbf{x}_R})\le \constnvsn |\monomial{p}{d'}{\sampledim'}|]
+p^{-\constnvsn |\monomial{p}{d'}{\sampledim'}|\cdot |L|}
\end{align*}
by (\ref{eq:rank}).
Observe that $t\ge \eta |\monomial{p}{d'}{\sampledim}|\ge \eta |\monomial{p}{d'}{\sampledim'}|$ so
we can apply Lemma~\ref{lem:rank} with $\sampledim=\sampledim'$, $d=d'$, and
$\mathbf{x}=\mathbf{x}_R$ to derive that
$$\Pr_{\mathbf{x_R}}[rank(Q_{\mathbf{x}_R})\le \constnvsn |\monomial{p}{d'}{\sampledim'}|]
p^{-c|\monomial{p}{d'+1}{\sampledim'}|}.$$
Therefore, 
\begin{equation}
\E_{f\in_R \Pdm}[\;|\biasj(f)|^{2t}\;]\le
p^{-c|\monomial{p}{d}{\sampledim'}|} +p^{-\constnvsn |\monomial{p}{d-1}{\sampledim'}|\cdot |L|}.
\label{eq:monomial}
\end{equation}
Now, for sufficiently large $\sampledim$, by Proposition~\ref{prop:d-vs-d-1}(a),
$|\monomial{p}{d}{\sampledim'}|\ge \gamma'|\monomial{p}{d}{\sampledim}|$ and by 
Corollary~\ref{prop:mprime-bound}
$|\monomial{p}{d-1}{\sampledim'}|\cdot |L|\ge \constnvsn|\monomial{p}{d}{\sampledim}|$.
Therefore,
$$\E_{f\in_R \Pdm}[\;|\biasj(f)|^{2t}\;]\le p^{-c\gamma' |\monomial{p}{d}{\sampledim}|}+
p^{-\constnvsn^2 |\monomial{p}{d}{\sampledim}}\ge
p^{-c'|\monomial{p}{d}{\sampledim}|}$$
for some constant $c'>0$.  
Now we can apply Markov's inequality to obtain that for any $c_1>0$.
\begin{align*}
\Pr_{f\in_R\Pdm}[|\biasj(f)|>p^{-c_1\sampledim/d}]&=\Pr_{f\in_R\Pdm}[|\biasj(f)|^{2t}>p^{-2t\cdot c_1\sampledim/d}]\\
&\leq \frac{p^{-c'|\monomial{p}{d}{\sampledim}|}}{p^{-2t\cdot c_1\sampledim/d}}\\
&=p^{2t\cdot c_1\sampledim/d-c'|\monomial{p}{d}{\sampledim}|}
\end{align*}
By definition, $t=\lceil \eta |\monomial{p}{d-1}{\sampledim}|\rceil\le \eta' |\monomial{p}{d-1}{\sampledim}|$ for a fixed $\eta'>\eta$.
Therefore, by Proposition \ref{prop:d-vs-d-1}(b),
$2t \sampledim/d\le  2\eta'\rho_2 |\monomial{p}{d}{\sampledim}|$.
By choosing $c_1=c'/(4\eta'\rho_2)$, we obtain that 
$2t\cdot c_1\sampledim/d-c'|\monomial{p}{d}{\sampledim}|\le -c'|\monomial{p}{d}{\sampledim}|/2$
and setting $c_2=c'/2$ we derive that
$$\Pr_{f\in_R\Pdm}[|\biasj(f)|>p^{-c_1\sampledim/d}]\le p^{-c_2|\monomial{p}{d}{\sampledim}|}$$
as required.
\end{proof}

It remains to prove Lemma~\ref{lem:rank} and Proposition~\ref{prop:d-vs-d-1}.
We first prove Lemma~\ref{lem:rank} using Lemma \ref{prop:evaluation-rank}.
The proof of Lemma~\ref{prop:evaluation-rank} is quite involved and forms the 
bulk of the paper.   Its proof is the subsequent sections.

\begin{proof}[Proof of Lemma~\ref{lem:rank} using Lemma~\ref{prop:evaluation-rank}]
Let $d=\floor{\delta \sampledim}$ for $0<\delta\le 1/2$ and let $\gamma>0$ be the
minimum of
$\gamma'(\delta)$ from Proposition~\ref{prop:d-vs-d-1} and $\gamma(\delta)$
from Corollary~\ref{prop:mprime-bound}.
Fix $b=\lfloor \constnvsn \cdot |\monomial{p}{d}{\sampledim}|\rfloor$.
We will first check the probability that an arbitrary fixed set of $b$ columns
spans the whole matrix, and then apply a union bound to obtain the final result.

Let $V$ denote the linear space spanned by those $b$ columns.
Recall that each column of $Q_{x_R}$ is the evaluation 
of all monomials of degree at most $d$ at some point $\F_p^\sampledim$.
(Since $d\ge 1$, distinct elements of $\F_p^\sampledim$ have distinct evaluations.)

Let integer $r$ be maximal such that there are at least $p^r$ distinct elements of $\F_p^\sampledim$ with evaluations that are in $V$.
Then by Lemma~\ref{prop:evaluation-rank},
we have $\dim(V)\geq |\monomial{p}{d}{r}|$.
But since $V$ can be spanned by $b$ vectors,
we have 
\[
\constnvsn |\monomial{p}{d}{\sampledim}|\geq b \geq \dim(V)\geq |\monomial{p}{d}{r}|
\]
By Proposition \ref{prop:d-vs-d-1}(a),
we have
\[
|\monomial{p}{d}{\lceil \sampledim(1-1/d)\rceil}|\geq \constnvsn |\monomial{p}{d}{\sampledim}|\geq |\monomial{p}{d}{r}|
\]
So $r\leq \lceil \sampledim(1-1/d) \rceil$.
There are $p^\sampledim$ distinct evaluations
and fewer than $p^{r+1}$ of them fall into $V$.
So a uniform random evaluation is in $V$ with probability 
$<\frac{p^{r+1}}{p^\sampledim}\leq p^{1-\floor{\sampledim/d}}$.
Since the $t-b$ other columns of $Q_{x_R}$ are chosen uniformly and
independently, the probability that these $b$ columns span the whole matrix is
at most 
\[
(p^{1-\floor{\sampledim/d}})^{t-b}\leq (p^{1-\floor{\sampledim/d}})^{(\eta-\constnvsn)|\monomial{p}{d}{\sampledim}|}
\]
since $t=\eta |\monomial{p}{d}{\sampledim}|$ for some $\eta>1$ to be chosen later.
Since $d\le \delta \sampledim\le \sampledim/2$, we have
$1-\floor{\sampledim/d}\le -\sampledim/(2d)$ and
we can apply Proposition \ref{prop:d-vs-d-1} to get that 
\[
(p^{1-\floor{\sampledim/d}})^{t-b}\leq p^{-(\eta-\constnvsn) \frac \sampledim d\cdot|\monomial{p}{d}{\sampledim}|/2}
\leq p^{-(\eta-\constnvsn)\rho_1 |\monomial{p}{d+1}{\sampledim}|/2}
\]
for some $\rho_1>0$.
Therefore,
by a union bound over all choices of $b$ columns we have
\[
\Pr_{x\ind{1},\cdots,x\ind{t}}[rank(Q_{x_R})\leq \constnvsn |\monomial{p}{d}{\sampledim}|]\leq \binom{t}{b} \cdot p^{-(\eta-\constnvsn) \rho_1 |\monomial{p}{d+1}{\sampledim}|/2}.
\]
Note that $\binom{t}{b}\leq (\frac{te}{b})^b\leq (\frac{2e\eta}{\gamma} )^{\constnvsn |\monomial{p}{d}{\sampledim}|}\le (\frac{2e\eta}{\gamma})^{\constnvsn |\monomial{p}{d+1}{\sampledim}|}$,
so we have
\[
\Pr_{x\ind{1},\cdots,x\ind{t}}[rank(Q_{x_R})\leq \constnvsn |\monomial{p}{d}{\sampledim}|]
\leq p^{|\monomial{p}{d+1}{\sampledim}|(\constnvsn\log_p (\frac{2e\eta}{\gamma})-(\eta-\constnvsn) \rho_1/2)}
\]
Note that for any constant $c'>0$, $\constnvsn\log_p  (c'\eta)$ is $o(\eta)$.
Therefore, for fixed constant $\gamma>0$,
we can choose a sufficiently large $\eta>1$ such that
\[
\Pr_{x\ind{1},\cdots,x\ind{t}}[rank(Q_{x_R})\leq \constnvsn |\monomial{p}{d}{\sampledim}|]\leq p^{-c|\monomial{p}{d+1}{\sampledim}|}
\]
for some constant $c>0$.
\end{proof}

\subsection{Proof of Proposition~\ref{prop:d-vs-d-1}}

We first give basic inequalities regarding
$|\monomial{p}{d}{\sampledim}|$ that are independent of the choice of $p$.

\begin{proposition}\label{prop:monomial-bound}
For $d\leq \sampledim$,
\[
\sum_{i=0}^{d}\binom{\sampledim}{i}\leq |\monomial{p}{d}{\sampledim}|\leq \sum_{i=0}^{d}\binom{\sampledim-1+i}{i}=\binom{\sampledim+d}{d}
\]
\end{proposition}

\begin{proof}
It is well known that there are $\binom{\sampledim+d-1}{\sampledim-1}$ non-negative integer solutions to the equation $\sum_{i=1}^\sampledim e_i= d$.
Thus by iterating degrees we have
\[
|\monomial{p}{d}{\sampledim}|\leq \sum_{i=0}^{d}\binom{\sampledim-1+i}{i}=\binom{\sampledim+d}{d}
\]
On the other hand,
if we only consider multi-linear terms,
we will get 
\[
\sum_{i=0}^{d}\binom{\sampledim}{i}\leq |\monomial{p}{d}{\sampledim}|.
\]
\end{proof}

We now prove part (a):
For $\mathbf{e}=(e_1,\cdots,e_k)$ where $1\leq e_i\leq p-1$,
let $\mathcal{M}_{\mathbf{e},\sampledim}$ denote the set of monomials of the form $\prod_{i=1}^k x_{h(i)}^{e_i}$,
$1\leq h(1)<h(2)<\cdots<h(k)\leq \sampledim$.
Then we have $|\monomial{p}{d}{\sampledim}|=\sum_{\mathbf{e}:\sum_i e_i\leq d} |\mathcal{M}_{\mathbf{e},\sampledim}|$.
Therefore
\[
\frac{|\monomial{p}{d}{\sampledim'}|}{|\monomial{p}{d}{\sampledim}|}\geq \min_{\mathbf{e}:\sum_i e_i\leq d}\frac{|\mathcal{M}_{\mathbf{e},\sampledim'}|}{|\mathcal{M}_{\mathbf{e},\sampledim}|}
\]
Now,
for fixed $\mathbf{e}=(e_1,\cdots,e_k)$,
consider the following process to generate elements in $\mathcal{M}_{\mathbf{e},\sampledim}$:
we first choose $k$ elements $j_1,\cdots, j_k$ from $[\sampledim]$,
then apply permutation $\phi$ over $[k]$ to get $\prod_{i=1}^k x_{j_{\phi(k)}}^{e_i}$.
We claim that this process generate each monomial in $\mathcal{M}_{\mathbf{e},\sampledim}$ equally many times,
if we go over all $k$ elements and all permutations.
Indeed,
for arbitrary monomials $f_1=\prod_{i=1}^k x_{j_i}^{e_i}$ and $f_2=\prod_{i=1}^k x_{j_i'}^{e_i}$,
$f_1$ can be generated by $(\{j_i\}_{i=1}^k,\phi)$ if and only if $f_2$ can be generated by $(\{j_i'\}_{i=1}^k,\phi)$.
Moreover,
the number of occurrence for each monomial is precisely the number of satisfying permutations,
hence only depends on $\mathbf{e}$.
Therefore,
we have 
\[
\frac{|\mathcal{M}_{\mathbf{e},\sampledim'}|}{|\mathcal{M}_{\mathbf{e},\sampledim}|}=\frac{\binom{\sampledim'}{k}}{\binom{\sampledim}{k}}
\]
This quantity is a decreasing function of $k$.
Hence we have
\[
\frac{|\monomial{p}{d}{\sampledim'}|}{|\monomial{p}{d}{\sampledim}|}\geq \frac{\binom{\sampledim'}{d}}{\binom{\sampledim}{d}}=\prod_{i=0}^{d-1}\frac{\sampledim'-i}{\sampledim-i}\geq (\frac{\sampledim'-d+1}{\sampledim-d+1})^d\geq e^{-\frac{(\sampledim-\sampledim')d}{\sampledim'-d+1}}\geq e^{-\frac {\sampledim}{\sampledim-\frac \sampledim {\sampledim'}-\sampledim'+1}}
\]
Since $d+\frac \sampledim d\leq \max\{2+\frac \sampledim 2,\delta \sampledim+\frac 1 {\delta}\}$,
we have $\sampledim-\frac \sampledim d-d+1\geq \min\{\frac \sampledim 2-1,(1-\delta)\sampledim+\frac 1 \delta+1\}$.
Therefore for sufficiently large $\sampledim$,
\[
\frac{|\monomial{p}{d}{\sampledim'}|}{|\monomial{p}{d}{\sampledim}|}\geq \min\{e^{-3},e^{-\frac 1 {1-\delta}}\}
\]

We now prove part (b):
Define $H:=\{(q_1,q_2):q_1\in \monomial{p}{d-1}{\sampledim},q_2\in \monomial{p}{d}{\sampledim},\exists i\in[\sampledim] \text{ s.t. } q_2=x_iq_1\}$.
We will obtain the inequalities by bounding $|H|$.

We first bound $|H|$ in terms of $|\monomial{p}{d-1}{\sampledim}|$.
Clearly for each $q_1\in \monomial{p}{d-1}{\sampledim}$,
there are at most $\sampledim$ choices of $x_i$ to yield $q_2=x_iq_1$,
so $|H|\leq \sampledim|\monomial{p}{d-1}{\sampledim}|$.
On the other hand, any $x_i$ that does not have degree 
$p-1$ in $q_1$ can be chosen.  There are at most $\frac{d-1}{p-1}$ variables
in $q_1$ having degree $p-1$
so we can choose at least $\sampledim-\frac {d-1} {p-1}> \sampledim-\frac{\sampledim}{p-1}\ge \frac{\sampledim}{2}$
variables $x_i$ since $p\ge 3$.
This gives us $|H|> \frac \sampledim 2 |\monomial{p}{d-1}{\sampledim}|$.

We now bound $|H|$ in terms of  $|\monomial{p}{d}{\sampledim}|$.
Each $q\in \monomial{p}{d}{\sampledim}$ contains at most $d$ distinct variables,
hence $|H|\leq d |\monomial{p}{d}{\sampledim}|$.

It immediately follows that $\monomial{p}{d}{\sampledim} \ge |H|/d \ge \frac{\sampledim}{2d} |\monomial{p}{d-1}|$ and hence we can choose $\rho_2=2$.

To lower bound $|H|$ in terms of $|\monomial{p}{d}{\sampledim}|$,
we show that a large portion of monomials contain many distinct variables and
hence each
$q_2\in\monomial{p}{d}{\sampledim}$ can be associated with many different $q_1$.
We first bound the number of monomials that have degree at most $d$ and
are composed of at most $k\le d$ distinct variables.
We can generate such monomials by first choosing $k$ variables,
then using these variables to form a monomial of degree $\le d$ and
so we can upper bound the number of such monomials by
$\binom{\sampledim}{k}\binom{k+d}{d}$.
For sufficiently small $k$ we can argue that this is a small fraction of 
$\monomial{p}{d}{\sampledim}$:
Suppose that $k\le d/6$.  Since by hypothesis, $d\le \delta \sampledim\le  \sampledim/2$, we have
$k+d\le \sampledim-k$ and
\begin{align*}
\frac{\binom{\sampledim}{k}\binom{k+d}{d}}{|\monomial{p}{d}{\sampledim}|}
&\le \frac{\binom{\sampledim}{k}\binom{k+d}{d}}{\binom{\sampledim}{d}}\qquad\mbox{by Proposition~\ref{prop:monomial-bound}}\\
&=\frac{(k+d)!}{(k!)^2 (\sampledim-k)\cdots (\sampledim-d+1)}
=\frac{(k+d)\cdots (2k+1)}{(\sampledim-k)\cdots (\sampledim-d+1)}\cdot \binom{2k}{k}\\
&\leq \left(\frac {k+d} {\sampledim-k}\right)^{d-k}\cdot 2^{2k}
\le (7/11)^{5k}\cdot 2^{2k} \le (3/7)^k. 
\end{align*}
Summing over all values of $k\le d/6$ we obtain that a total fraction at most
3/4 of all monomials in $\monomial{p}{d}{\sampledim}$ have at most $d/6$ distinct
variables.  
Therefore, since at least 1/4 of $\monomial{p}{d}{\sampledim}$ contain at least 
$d/6$ distinct variables,
it must be the case that $|H|\geq \frac d 24\cdot |\monomial{p}{d}{\sampledim}|$.
Since $|H|\le \sampledim\cdot |\monomial{p}{d-1}{\sampledim}|$, we obtain that 
$|\monomial{p}{d}{\sampledim}|/24\le \frac{\sampledim}{d} |\monomial{p}{d-1}{\sampledim}|$. Hence
we derive (b) with $\rho_1=1/24$.
This completes the proof of Proposition~\ref{prop:d-vs-d-1}(b).

\subsection{Lower bound on the likelihood of bias}

We now prove Proposition~\ref{prop:lower-bias}, 
on the limits on the extent to which
Theorem~\ref{lem:main-bias} can be improved.
The argument is analogous to that
of~\cite{DBLP:journals/cc/Ben-EliezerHL12} for the case of
$\mathbb{F}_2$.

\begin{proof}[Proof of Proposition~\ref{prop:lower-bias}]
We follow the same division of variables $[\sampledim]$ into parts
$L$ and $R$ with $|L|=\lfloor\frac{n}{d}\rfloor$ and 
$|R|=\sampledim'=\lceil \sampledim(1-1/d)\rceil$ and $d'=d-1$ that
we used for the upper bound on the bias.
Define $\mathcal{L}$ to be the set of all polynomials in $\Pdm$ whose
monomials are from the set 
$\mathcal{M}'\subseteq \monomial{p}{d}{\sampledim}$ (defined earlier)
that have degree 1 on $L$ and degree at most $d-1$ on $R$.
By Corollary~\ref{prop:mprime-bound}, there is some constant
$\gamma>0$ such that for sufficiently large $\sampledim$,
$|\mathcal{M}'|\ge \gamma\cdot |\monomial{p}{d}{\sampledim}|$ and
hence $|\mathcal{L}|\ge p^{\gamma |\monomial{p}{d}{\sampledim}|}$.
Therefore, we have
$|\mathcal{L}|/|\Pdm|\ge p^{-(1-\gamma)|\monomial{p}{d}{\sampledim}|}$.

Now consider the expected bias of polynomials in $\mathcal{L}$:
We can write $f$ chosen uniformly from $\mathcal{L}$ uniquely as
$$f(x)=\sum_{i\in L} x_i\cdot g_i(x_R)$$
where the $g_i$ are independently chosen polynomials over monomials 
$\monomial{p}{d-1}{\sampledim'}$ on $R$.

For $j\in \mathbb{F}^*_p$,
\[\begin{split}
\E_{f\in_R \mathcal{L}}\ \biasj(f)=&
\E_{f\in_R \mathcal{L}}\ \E_{x\in_R \mathbb{F}_p^\sampledim}\  \omega^{j\cdot f(x)}\\
=&\E_{f\in_R \mathcal{L}}\ 
\E_{x_L\in_R \mathbb{F}_p^L}\ \E_{x_R\in_R \mathbb{F}_p^R }\ 
\omega^{j\cdot f(x)}\\
=&\E_{x_L\in_R \mathbb{F}_p^L}
\E_{x_R\in_R \mathbb{F}_p^R}\E_{f\in_R \mathcal{L}}\ \omega^{j\cdot f(x)}.
\end{split}
\]
Now with probability $p^{-|L|}$, all the $x_i$ for $i\in L$ are 0
and every $f\in \mathcal{L}$ evaluates to 0, so $\E_{f\in_R \mathcal{L}}\ \omega^{j\cdot f(0^L,x_R)}=1$.
With the remaining probability, $x_L\ne 0^L$ and hence there is some 
$i \in L$ and $b_i\ne 0$ such that $x_i=b_i$.  
For $f$ chosen at random from $\mathcal{L}$, for each fixed value of 
$x_L=b_L$ with $b_i\ne 0$, we have 
$$f(b_L,x_R)= b_i g_{i0} + f'(x_R)$$ where
$g_{i0}$ is the constant term of the polynomial $g_i$ and is chosen
independently of $f'$. Since $g_{i0}$ is uniformly chosen from 
$\mathbb{F}_p$ for random $f$ in $\mathcal{L}$ 
and since $b_i\ne 0$,
$b_i g_{i0}$ is also uniformly chosen from $\mathbb{F}_p$.  
Further,
since $g_{i0}$ is independent of $f'$, for every fixed $x_R$,
the value $\E_{f\in_R \mathcal{L}} \omega^{j\cdot f(b_L,x_R)}=0$.
Therefore,
$\E_{f\in_R \mathcal{L}} \biasj(f)=p^{-|L|}$.
Now since $|\biasj(f)|\le 1$, 
we obtain
$$\Pr_{f\in_R \mathcal{L}}[\;|\biasj(f)| \ge p^{-|L|}/2\;]\ge p^{-|L|}/2.$$
Therefore 
$$\Pr_{f\in_R \Pdm} [|\biasj(f)|\ge p^{-|L|}/2] \ge \frac{|\mathcal{L}|}{|\Pdm|}\cdot p^{-|L|}/2\ge p^{-c'|\monomial{p}{d}{\sampledim}|}$$
for some $c'<1$ since $|L|\ll
|\monomial{p}{d}{\sampledim}|$.
Since $|L|=\lfloor n/d \rfloor\ge 2$, we obtain $p^{-|L|}/2>p^{-c''n/d}$
for some constant $c''>0$ as required.
\end{proof}

\section{Extremal rank properties of truncated Reed-Muller codes}
\label{sec:reed-muller}
In this section we prove Lemma \ref{prop:evaluation-rank}.
Let $\generatingmatrix{d}$ be the natural generating matrix of the $(d,\sampledim)$ Reed-Muller code over the field $\mathbb{F}_p$ for $p$ an odd prime.
That is,
$\generatingmatrix{d}$ is a $p^\sampledim \times |\monomial{p}{d}{\sampledim}|$ matrix,
where each row is indexed by $a\in \mathbb{F}_p$,
and each column is indexed by a monomial $q\in \monomial{p}{d}{\sampledim}$,
and $M^{(d)}(x,q)=q(x)$.
Clearly the rank of $\generatingmatrix{d}$ is  $|\monomial{p}{d}{\sampledim}|$,
since all the columns are independent.
For $S\subset \mathbb{F}_p^\sampledim$,
define $\generatingmatrix{d}_S$ as $\generatingmatrix{d}$ restricted on rows in $S$.

For fixed $q\in \monomial{p}{d}{\sampledim}$,
$(q(x))_{x\in S}$ is precisely the column corresponding to $q$ in $M^{(d)}$ restricted to  rows indexed by $S$. From this prospective,
the dimension of the subspace spanned by $\{(q(x))_{q\in \monomial{p}{d}{\sampledim}}:x\in S\}$ is exactly the (column) rank of the matrix $\generatingmatrix{d}_S$.
So we can restate Lemma \ref{prop:evaluation-rank} as the following:

\begin{lemma}\label{prop:evalutation-rank-new}
Let $S$ be a subset of $\mathbb{F}_p^\sampledim$ such that $|S|=p^r$,
then $rank(\generatingmatrix{d}_S)\geq |\monomial{p}{d}{r}|$.
\end{lemma}

As noted in the introduction we will derive the more general bound for an
arbitrary value of $|S|$, not only for the restricted case that $|S|=p^r$ that occurs in the above lemma.
We start with the special case where $S$ contains the lexicographically smallest  $|S|$ elements. 
(Here lexicographical order means that for $a=(a_1,\cdots,a_\sampledim)$ and $b=(b_1,\cdots,b_\sampledim)$,
$a<b$ if and only if there exists $1\leq k\leq \sampledim$ such that $a_i=b_i$ for $i<k$ and $a_k<b_k$.)
We will show that this is in fact an extremal case.

For integer $\iterator$,
let $S_\iterator\subset \mathbb{F}_p^\sampledim$ be the subset that contains the smallest $\iterator$ elements.
Define $g_d(\iterator)$ as the rank of $\generatingmatrix{d}_{S_\iterator}$.
For the completeness of definition,
we set $g_d(\iterator)=0$ when $d<0$ or $\iterator=0$.
When $\iterator$ is a power of $p$,
it is easy to compute the value of $g_d(\iterator)$.

\begin{lemma}\label{thm-g_d-base}
$
g_d(p^r)=|\monomial{p}{d}{r}|.
$
\end{lemma}

\begin{proof}
Let $x$ be a monomial of degree at most $d$ over $x_1,\cdots,x_\sampledim$.
Notice that $\forall a\in S_{p^r}$,
$a_1=a_2=\cdots=a_{\sampledim-r}=0$.
So if $x$ contains any of the first $\sampledim-r$ variables,
the column indexed by $x$ will be $0$.
On other other hand,
all the columns corresponding to monomials over $x_{\sampledim-r+1},\cdots,x_\sampledim$ of degree at most $d$ are linear independent.
Otherwise this means that some monomial can be represented as a linear combination of other monomials, which is impossible.
So the rank of the submatrix is just the number of monomials over $r$ variables of degree at most $d$.
\end{proof}

It is more complicated to compute the value for general $\iterator$.
With careful observation,
we have the following recursion.

\begin{lemma}\label{def-g_d}
Let $r$ be the unique integer so that $p^r\leq \iterator<p^{r+1}$.
Let $\iterator=k\cdot p^r+c$.
Then
\[
g_d(\iterator)=\sum_{i=0}^{k-1}g_{d-i}(p^r)+g_{d-k}(c).
\]
As a special case, when $c=0$, we have
$
g_d(k\cdot p^r)=\sum_{i=0}^{k-1}g_{d-i}(p^r)
$.
\end{lemma}

\begin{proof}
For the sake of convenience,
let $\monomialeq{p}{d}{r}$ be the set of monomials over the last $r$ variables whose degree equals $d$,
and  $\monomialle{p}{d}{r}$ be the set of monomials over the last $r$ variables whose degree is at most $d$.

Consider the block structure of the matrix.
Let $A_0$ be the submatrix that takes $S_{p^r}$ as rows and  $\monomialle{p}{d-p+1}{r}$ as columns.
For $i=1,2,\cdots,p-1$,
let $A_i$ be the submatrix that takes $S_{p^r}$ as rows and $\monomialeq{p}{d-p+1+i}{r}$ as columns.
Let $A_{\leq i}$ be the blocks $(A_0,\cdots,A_i)$,
then $A_{\leq i}$ is precisely the submatrix that takes $S_{p^r}$ as rows and $\monomialle{p}{d-p+i+1}{r}$ as columns.

Now let us consider the rows for $R_t:=\{a|a_1=\cdots=a_{\sampledim-r-1}=0,a_{\sampledim-r}=t\}$ for $t>0$.
The non-zero parts correspond to monomials that only depend on $x_{\sampledim-r},x_{\sampledim-r+1},\cdots,x_{\sampledim}$.
If we group all the monomials by their degree on $x_{\sampledim-r}$
then, for $t\leq k-1$,
the row will be of the form
\[
A_{\leq p-1},\,t\cdot A_{\leq p-2},\, t^2\cdot A_{\leq p-3},\,
\cdots, \,t^{p-1}\cdot A_{\leq 0}.
\]
Things are a little different for $t=k$,
since $|R_k|<p^r$.
In this case,
we define $A_{\leq i}'$ as the first $c$ rows of $A_{\leq i}$,
and it is easy to check that the row is of the form
\[
A_{\leq p-1}',\,k\cdot A_{\leq p-2}',\, k^2\cdot A_{\leq p-3}',\,
\cdots, \,k^{p-1}\cdot A_{\leq 0}'.
\]
Therefore the matrix $\generatingmatrix{d}_{S_\iterator}$ is of the form:

\begin{center}
\begin{tabular}{ c||c|c|c|c|c } 
& $\monomialle{p}{d}{r} $ &$x_{\sampledim-r}\cdot \monomialle{p}{d-1}{r}$ &$x_{\sampledim-r}^2\cdot \monomialle{p}{d-2}{r}$&$\cdots$ & $x_{\sampledim-r}^{p-1}\cdot  \monomialle{p}{d-p+1}{r}$\\
\hline
\hline
 $\{a_{\sampledim-r}=0\}$ & $A_{\leq p-1}$ & $0$ & $0$ & $\cdots$ & $0$  \\  
 \hline
   $\{a_{\sampledim-r}=1\}$ & $A_{\leq p-1}$ & $A_{\leq p-2}$ & $A_{\leq p-3}$ & $\cdots$ & $A_{\leq 0}$  \\   
 \hline
   $\{a_{\sampledim-r}=2\}$ & $A_{\leq p-1}$ & $2\cdot A_{\leq p-2}$ & $4\cdot A_{\leq p-3}$ & $\cdots$ & $2^{p-1}\cdot A_{\leq 0}$  \\   
 \hline
 \vdots & \vdots &\vdots & \vdots & \vdots &\vdots\\
 \hline
  $\{a_{\sampledim-r}=k-1\}$ & $ A_{\leq p-1}$ & $(k-1)\cdot A_{\leq p-2}$ & $(k-1)^2\cdot A_{\leq p-3}$ & $\cdots$ & $(k-1)^{p-1}\cdot A_{\leq 0}$\\
 \hline
  $\{a_{\sampledim-r}=k\}$ & $ A_{\leq p-1}'$ & $k\cdot A_{\leq p-2}'$ & $k^2\cdot A_{\leq p-3}'$ & $\cdots$ & $k^{p-1}\cdot A_{\leq 0}'$\\
 \hline
\end{tabular}
\end{center}

We have two important observations:
\begin{itemize}
\item $A_{\leq p-i-1}$ is the first $|\monomialle{p}{d-i-1}{r}|$ columns of $A_{\leq p-i}$, and
\item $A_{\leq i}'$ is the first $c$ rows of $A_{\leq i}$.
\end{itemize}
Therefore, we can apply Gaussian elimination to turn the matrix in to a block-diagonal matrix.
We do this in two steps.
The following algorithm first eliminates on columns to obtain a triangular matrix, based on the first observation.

\begin{algorithm}[H]\label{algorithm-column-elimination}
\caption{Triangular elimination}
Initialize $I=\{0,\cdots,k\}$, $u=0$.\\
For $i,j\in \{0,1,\cdots,k\}$,
set $b_{i,j}=i^j$.\\
\tcc{$b_{i,j}$ is the coefficient of each block at the beginning.}
\For{$u=0,\cdots,k$}{
\For{$v=u+1,\cdots,p-1$}{
\tcc{Subtract $b_{u,v}b_{u,u}^{-1}$ times the prefix of length $|\monomialle{p}{d-v}{r}|$ of  $u$-th column block from $v$-th column block. }
\For{$i\in \{0,1,\cdots,p-1\}$}{
$
b_{i,v}\leftarrow b_{i,v}-b_{u,v}b_{u,u}^{-1}b_{i,u}
$
}
}
}
\end{algorithm}

We have the following properties.
\begin{claim}\label{lem-elimination}
At the beginning of round $u$ of Algorithm \ref{algorithm-column-elimination}, 
for all $i\leq u$, we have
$b_{i ,i}\neq 0$,
and at the end of round $u$,
for all $i\leq u$, we have
$b_{i ,j}=0$ for $j=i+1,\cdots,p-1$.
\end{claim}

\begin{proof}
Let us consider the function $f^{(\rounditerator)}_j(i)$,
which denotes the coefficient of the $(i,j)$ block at the beginning of round $\rounditerator$.
We prove a strengthened claim: For $\rounditerator=0,\cdots,k+1$,
\begin{itemize}
\item $f^{(\rounditerator)}_j(i)$ is a monic degree $j$ polynomial on variable $i$,
\item $f^{(\rounditerator)}_i(i)\neq 0$ for $i\leq \rounditerator$, and
\item when $\rounditerator>0$,
$f^{(\rounditerator)}_j(i)=0$ for $i\leq \rounditerator-1$ and $j=i+1,\cdots,p-1$.
\end{itemize}

When $\rounditerator=0$,
we have $f^{(0)}_j(i)=i^j$ which is monic and has degree $j$.
Also,
$f^{(0)}_0(0)=1\neq 0$.
So the base case is true.

The update rule given by the algorithm says
\[
f^{(\rounditerator)}_j(i)=\begin{cases} f^{(\rounditerator-1)}_j(i)-\frac{f^{(\rounditerator-1)}_j({\rounditerator-1})}{f^{(\rounditerator-1)}_{\rounditerator-1}({\rounditerator-1})}\cdot f^{(\rounditerator-1)}_{\rounditerator-1}(i)&\mbox{if } j\geq \rounditerator\\
f^{(\rounditerator-1)}_j(i)& \mbox{if } j<\rounditerator
\end{cases}
\]
When $j<\rounditerator$, clearly $f^{(\rounditerator)}_j(i)$ is still a monic degree $j$ polynomial.
When $j\geq \rounditerator$,
by the induction hypothesis,
$f^{(\rounditerator-1)}_{\rounditerator-1}({\rounditerator-1})\neq 0$,
and the degree of $f^{(\rounditerator-1)}_{\rounditerator-1}(i)$ is strictly smaller than $j$.
Therefore we have that
$f^{(\rounditerator)}_j(i)$ is also a monic  degree $j$ polynomial.

For $i\leq t-2$,
by the induction hypothesis,
$f^{(t-1)}_j(i)=0$ for $j\geq i+1$.
Then both update rules give us $f^{(t)}_j(i)=0$ for $j\geq i+1$.
For the $i=t-1$ case,
 by Gaussian elimination,
$f^{(t)}_j(t-1)=0$ for all $j\geq t$.
Combining the two cases we have for $i\leq t-1$,
$f^{(t)}_j(i)=0$ for $j\geq i+1$.

Notice that
\[
f^{(\rounditerator)}_{\rounditerator}(0)=f^{(\rounditerator)}_{\rounditerator}(1)=\cdots=f^{(\rounditerator)}_{\rounditerator}({\rounditerator-1})=0.
\]
Since $f^{(\rounditerator)}_{\rounditerator}(i)$ is monic and has degree $\rounditerator$,
it can have at most $\rounditerator$ roots.
This implies that $\rounditerator$ cannot be its root,
hence $f^{(\rounditerator)}_{\rounditerator}(\rounditerator)\neq 0$.
\end{proof}

Using this claim we complete the proof of Lemma~\ref{def-g_d}: 
After the first step,
the matrix is in the form: 
\begin{center}
\scalebox{0.8}{%
\begin{tabular}{ c||c|c|c|c|c|c } 
& $\monomialle{p}{d}{r} $ &$x_{\sampledim-r}\cdot \monomialle{p}{d-1}{r}$ &$x_{\sampledim-r}^2\cdot \monomialle{p}{d-2}{r}$&$\cdots$ & $x_{\sampledim-r}^{k-1}\cdot  \monomialle{p}{d-k+1}{r}$ & $x_{\sampledim-r}^{k}\cdot  \monomialle{p}{d-k}{r}$  \\
\hline
\hline
 $\{a_{\sampledim-r}=0\}$ & $b_{0,0}A_{\leq p-1}$ & $0$ & $0$ & $\cdots$ & $0$ & $0$ \\  
 \hline
   $\{a_{\sampledim-r}=1\}$ & $b_{1,0}A_{\leq p-1}$ & $b_{1,1}A_{\leq p-2}$ & $0$ & $\cdots$ & $0$  & $0$\\   
 \hline
   $\{a_{\sampledim-r}=2\}$ & $b_{2,0}A_{\leq p-1}$ & $b_{2,1}A_{\leq p-2}$ & $ b_{2,2}A_{\leq p-3}$ & $\cdots$ & $0$ & $0$ \\   
 \hline
 \vdots & \vdots &\vdots & \vdots & \vdots &\vdots &\vdots \\
 \hline
  $\{a_{\sampledim-r}=k-1\}$ & $ b_{k-1,0}A_{\leq p-1}$ & $b_{k-1,1}A_{\leq p-2}$ & $b_{k-1,2}A_{\leq p-3}$ & $\cdots$ &$b_{k-1,k-1}A_{\leq p-k}$ & $0$\\
 \hline
  $\{a_{\sampledim-r}=k\}$ & $ b_{k,0}A_{\leq p-1}'$ & $b_{k,1}A_{\leq p-2}'$ & $b_{k,2}A_{\leq p-3}'$ & $\cdots$ & $b_{k,k-1}A_{\leq p-k}'$ & $ b_{k,k}A_{\leq p-k-1}'$\\
 \hline
\end{tabular}
}
\end{center}
with all other columns having value 0.
Now, the second observation says that $A_{\leq p-i}'$ is a submatrix of
$A_{\leq p-i}$.
Moreover,
we have $b_{i,i}\neq 0$ for $i=0,\cdots,k$.
So we can eliminate row by row
to get a diagonal matrix,
whose rank is very easy to compute.
The rank of $A_{\leq p-i}$ is simply $|\monomialle{p}{d-i+1}{r}|=g_{d-i+1}(p^r)$.
By the definition of the $g_d$ function,
the rank of $A_{\leq 0}'$ is just $g_{d-k}(c)$.
Hence
\[
g_d(\iterator)=g_d(k\cdot p^r+c)=\sum_{i=0}^{k-1}g_{d-i}(p^r)+g_{d-k}(c).\\[-3ex]
\]
\end{proof}

Intuitively,
for any set $S\subset p^\sampledim$ of size $\iterator$,
the rank of $\generatingmatrix{d}_S$ will be larger than that of $\generatingmatrix{d}_{S_\iterator}$,
since $S_\iterator$ is the most compact way to arrange $\iterator$ rows.
Formally stated, this is the following restatement of 
Theorem~\ref{thm:main-extremal}:

\begin{theorem}\label{rank-result}
For arbitrary $S\subseteq p^\sampledim$ with $|S|=\iterator$ and for any $d$,
\[
rank(\generatingmatrix{d}_S)\geq rank(\generatingmatrix{d}_{S_\iterator})=g_d(\iterator).
\]
\end{theorem}

Note that $g_d(\iterator)$ does not depend on $\sampledim$.
Indeed, all we need of $\sampledim$ is that $\iterator \le p^\sampledim$ so that the matrix
$\generatingmatrix{d}$ has at least $m$ rows.

Before we actually prove Theorem~\ref{rank-result},
we first argue that this is all we need to prove Lemma \ref{prop:evalutation-rank-new}.
Indeed,
we can simply set $\iterator=p^r$;
then with Lemma~\ref{thm-g_d-base},
we just have when $|S|=p^r$,
\[
rank(\generatingmatrix{d}_S)\geq g_d(p^r)=|\monomial{p}{d}{r}|.
\]

It turns out that in order to prove Theorem \ref{rank-result},
it is sufficient to have the following sub-additivity property of the $g_d$ function.
\begin{lemma}\label{lem-main}
For $a_0\geq a_1\geq \cdots \geq a_{p-1}\geq 0$,
for any $d$,
\[
g_d(\sum_{i=0}^{p-1}a_i)\leq \sum_{i=0}^{p-1}g_{d-i}(a_i).
\]
\end{lemma}

\begin{proof}[Proof of Theorem~\ref{rank-result} from Lemma \ref{lem-main}]
We use induction on the size of $|S|$. When $|S|=1$,
$rank(\generatingmatrix{d}_S)=1$,
while  $g_d(1)=1$,
so the base case is true.

Assume that we have proved that for all $S'$ with $|S'|\leq |S|$ and all degrees $d$,
$rank(\generatingmatrix{d}_{S'})\geq g_d(|S'|)$.
Let $t\in [\sampledim]$ be the smallest integer so that $\exists x,y\in S, x_t\neq y_t$.
For $i\in \{0,1,\cdots,p-1\}$,
define $S_i=\{x\in S:x_t=i\}$.
For $j=1,\cdots,p$,
let  $A_{\leq p-j}^{(i)}$ be the submatrix of $M$ that takes $S_i$ as rows and $\monomialle{p}{d+1-j}{p-t}$ as columns.
We claim that we only need to study the following matrix:
\begin{center}
\begin{tabular}{ c|c|c|c|c|c } 
& $\monomialle{p}{d}{\sampledim-t} $ &$x_{t}\cdot \monomialle{p}{d-1}{\sampledim-t}$ &$x_{t}^2\cdot \monomialle{p}{d-2}{\sampledim-t}$&$\cdots$ & $x_{t}^{p-1}\cdot  \monomialle{p}{d-p+1}{\sampledim-t}$\\
\hline
 $S_0$ & $A_{\leq p-1}^{(0)}$ & $0$ & $0$ & $\cdots$ & $0$  \\  
 \hline
   $S_1$ & $A_{\leq p-1}^{(1)}$ & $A_{\leq p-2}^{(1)}$ & $A_{\leq p-3}^{(1)}$ & $\cdots$ & $A_{\leq 0}^{(1)}$  \\   
 \hline
   $S_2$ & $A_{\leq p-1}^{(2)}$ & $2\cdot A_{\leq p-2}^{(2)}$ & $4\cdot A_{\leq p-3}^{(2)}$ & $\cdots$ & $2^{p-1}\cdot A_{\leq 0}^{(2)}$  \\   
 \hline
 \vdots & \vdots &\vdots & \vdots & \vdots &\vdots\\
 \hline
  $S_{p-2}$ & $ A_{\leq p-1}^{(p-2)}$ & $(p-2)\cdot A_{\leq p-2}^{(p-2)}$ & $(p-2)^2\cdot A_{\leq p-3}^{(p-2)}$ & $\cdots$ & $(p-2)^{p-1}\cdot A_{\leq 0}^{(p-2)}$\\
 \hline
  $S_{p-1}$ & $ A_{\leq p-1}^{(p-1)}$ & $(p-1)\cdot A_{\leq p-2}^{(p-1)}$ & $(p-1)^2\cdot A_{\leq p-3}^{(p-1)}$ & $\cdots$ & $(p-1)^{p-1}\cdot A_{\leq 0}^{(p-1)}$\\
 \hline
\end{tabular}
\end{center}
This is because all the other columns can be spanned by this matrix.
Indeed, consider a monomial $x=y\cdot z$ where $y$ is a monomial over $x_1,\cdots,x_{k-1}$ and $z$ is over $x_k,\cdots,x_\sampledim$.
Since for all $a\in S$, $a_i$ is fixed for $i=1,\cdots,k-1$,
the column for $x$ is just the column for $z$ times some constant.

We also have one important observation about this matrix.
\begin{itemize}
\item For all $i,j$, $A_{\leq p-j-1}^{(i)}$ is the first $|\monomialle{p}{d-j-1}{\sampledim-t}|$ columns of $A_{\leq p-j}^{(i)}$
\end{itemize}
so we can again use Gaussian elimination.
Although we may not be able to get a diagonal matrix because we do not know the relationship between the $S_i$,
we can carefully modify algorithm \ref{algorithm-column-elimination} to obtain a triangular matrix:

\begin{algorithm}[H]\label{algorithm-elimination}
\caption{Triangular elimination revised}
Initialize $I=\{0,\cdots,p-1\}$, $u=0$.\\
For $i,j\in \{0,1,\cdots,p-1\}$,
set $b_{i,j}=i^j$.\\
\tcc{$b_{i,j}$ is the coefficient of each block at the beginning.}
\While{$I\neq \emptyset$}
{
$i_0=\arg\max_{i\in I}|S_i|$.\\
\For{$v=u+1,\cdots,p-1$}{
\tcc{Subtract $b_{i_0,v}b_{i_0,u}^{-1}$ times $u$-th column block from $v$-th column block }
\For{$i\in \{0,1,\cdots,p-1\}$}{
$
b_{iv}\leftarrow b_{i,v}-b_{i_0,v}b_{i_0,u}^{-1}b_{i,u}
$
}
}
$I\leftarrow I-\{i_0\}$.\\
$u\leftarrow u+1$.
}
\end{algorithm}

Let $(\sigma_0,\sigma_1,\cdots,\sigma_{p-1})$ be the order of indices that 
Algorithm~\ref{algorithm-elimination} uses.
Then we have the following properties.
\begin{claim}\label{lem-elimination-new}
At the end of Algorithm 2 we have
\begin{itemize}
\item $b_{\sigma_i ,j}=0$ for $j=i+1,\cdots,p-1$ and
\item $b_{\sigma_i ,i}\neq 0$.
\end{itemize}
\end{claim}

The proof of Claim~\ref{lem-elimination-new} is very similar to that of Claim~\ref{lem-elimination} and we omit it here.
After the elimination,
the matrix is of the form:
\begin{center}
\begin{tabular}{ c|c|c|c|c|c } 
& $\monomialle{p}{d}{\sampledim-t} $ &$x_{t}\cdot \monomialle{p}{d-1}{\sampledim-t}$ &$x_{t}^2\cdot \monomialle{p}{d-2}{\sampledim-t}$&$\cdots$ & $x_{t}^{p-1}\cdot  \monomialle{p}{d-p+1}{\sampledim-t}$\\
\hline
 $S_{\sigma_0}$ & $b_{\sigma_0,0}\cdot A_{\leq p-1}^{(\sigma_0)}$ & $0$ & $0$ & $\cdots$ & $0$  \\  
 \hline
   $S_{\sigma_1}$ &  & $b_{\sigma_1,1}\cdot A_{\leq p-2}^{(\sigma_1)}$ & 0 & $\cdots$ & 0  \\   
 \hline
   $S_{\sigma_2}$ &  & & $b_{\sigma_2,2}\cdot A_{\leq p-3}^{(\sigma_2)}$ & $\cdots$ & 0  \\   
 \hline
 \vdots & \vdots &\vdots & \vdots & \vdots &\vdots\\
 \hline

  $S_{\sigma_{p-1}}$ &&  &  & $\cdots$ & $b_{\sigma_{p-1},(p-1)}\cdot A_{\leq 0}^{(\sigma_{p-1})}$\\
 \hline
\end{tabular}
\end{center}
Since it is a triangular matrix,
we can lower bound its rank as
\[
\sum_{i=0}^{p-1}rank(A_{\leq p-1-i}^{(\sigma_i)}).
\]
Recall that  columns of $A_{\leq p-1-i}^{(\sigma_i)}$ are $\monomialle{p}{d-i}{\sampledim-t}$,
so $A_{\leq p-1-i}^{(\sigma_i)}$ is actually $M^{(d-i)}_{S_{\sigma_i}}$.
Since $|S_{\sigma_i}|<|S|$,
 we can apply the induction hypothesis to get
\[
rank(M^{(d-i)}_{S_{\sigma_i}})\geq g_{d-i}(|S_{\sigma_i}|).
\]
Now we can apply Lemma \ref{lem-main} to get
\[
rank(M^{(d)}_S)\geq \sum_{i=0}^{p-1}g_{d-i}(|S_{\sigma_i}|)\geq g_d(\sum_{i=0}^{p-1}|S_{\sigma_i}|)=g_d(\iterator)
\]
as required.
\end{proof}

Therefore to complete our proof of the extremal rank properties of these matrices
in Theorem~\ref{rank-result}, and hence Theorem{thm:main-extremal} and Lemma~\ref{prop:evaluation-rank}, it only remains to prove the sub-additivity property of Lemma~\ref{lem-main} for the $g_d$ function.

\section{Properties of the $g_d$ function}
\label{sec:g-d-property}

In this section, we generalized the $g_d$ function and make some observations about its properties.  We then give the proof of Lemma~\ref{lem-main} in the next section.

From Theorem \ref{def-g_d},
we can see that it is very easy to compute $g_d(\iterator)$ if we write $\iterator$ in base $p$.
It will be convenient to consider a more general class of functions in other bases.
\begin{defn}\label{def-g_d,q}
Fix $d\in \mathbb{Z}$ and $q\in \mathbb{Z}_{>0}$.
Let $\func{d}{q}:\mathbb{Z}\rightarrow \mathbb{Z}$ be a function defined as
$\func{q}{d}(\iterator)=0$ when $d<0$ or $\iterator=0$,
$\func{q}{d}(1)=1$ when $d\geq 0$,
$\func{d}{q}(q^r)=\sum_{i=0}^{q-1}\func{d-i}{q}(q^{r-1})$,
and for other cases,
let $r$ be the largest integer so that $\iterator\geq q^r$,
then
\[
\func{d}{q}(\iterator)=\func{d}{q}(q^r)+\func{d-1}{q}(\iterator-q^r)
\]
\end{defn}

We can verify that $g_d(\iterator)$ defined in last section is the same as $\func{d}{p}(\iterator)$.
\begin{proposition}\label{prop-expand}
$
g_d(\iterator)=\func{d}{p}(\iterator).
$
\end{proposition}
\begin{proof}
When $d<0$ or $\iterator=0$,
by definition of $g_d$, we have
$g_d(\iterator)=0$.
Also,
when $d\geq 0$,
by Theorem~\ref{thm-g_d-base}, we have
$g_d(1)=|\monomial{p}{d}{0}|=1$.

Next, we claim that $g_d(p^r)=\sum_{i=0}^{p-1}g_{d-i}(p^{r-1})$.
Since $g_{d}(p^r)=|\monomial{p}{d}{r}|$,
we can split all monomials over $x_1,\cdots,x_r$ of degree at most $d$ by their degree in $x_1$.
That is,
\[
\monomial{q}{d}{r}=\bigcup_{i=0}^{q-1} x_1^i\cdot \monomial{q}{d-i}{r-1},
\]
which precisely yields the claim.

Finally,
for input $\iterator$ not of the above forms,
let $r$ be the largest integer so that $\iterator\geq q^r$,
then Theorem \ref{def-g_d} implies that
\[
g_d(\iterator)-g_{d-1}(\iterator-p^r)=g_d(p^r).
\]
So, $g_d(\iterator)$ satisfies all the requirements for $\func{d}{p}(\iterator)$.
Notice that for fixed $p$,
Definition \ref{def-g_d,q} uniquely determines a function over $\mathbb{Z}$,
so we have $g_d(\iterator)=\func{d}{p}(\iterator)$.
\end{proof}

We prove the following important property of the $\func{d}{q}$ function that
we will repeatedly apply later.

\begin{proposition}\label{prop-main}
For $d>d'$ and $q,a> 0$,
\[
\func{d}{q}(a)-\func{d}{q}(a-1)\geq \func{d'}{q}(a)-\func{d'}{q}(a-1).
\]
\end{proposition}

\begin{proof}
We first write down $\func{d}{q}(a+1)-\func{d}{q}(a)$ as a function of $d$.
We represent $a$ and $a-1$, respectively, in base $q$ as
\[
a=\sum_{i=0}^r a_i\cdot q^i \quad\mbox{ and } \quad  a-1=\sum_{i=0}^r b_i\cdot q^i.
\]
With this representation,
by Definition \ref{def-g_d,q}
we can explicitly compute $$\func{d}{q}(a)=\sum_{i=0}^{r}\func{d-\sum_{j=i+1}^{r}a_j}{q}(a_i\cdot q^r)$$
Let $i_0$ be the largest integer so that $a_{i_0}\neq b_{i_0}$.
Notice that
\[
\func{d}{q}(a)-\func{d}{q}(a-1)=\sum_{i=0}^{r}\func{d-\sum_{j=i+1}^{r}a_j}{q}(a_i\cdot q^r)-\sum_{i=0}^{r}\func{d-\sum_{j=i+1}^{r}b_j}{q}(b_i\cdot q^r)
\]
When $i>i_0$, $a_i=b_i$,
so we have $\func{d-\sum_{j=i+1}^{r}a_j}{q}(a_i\cdot q^r)=\func{d-\sum_{j=i+1}^{r}b_j}{q}(b_i\cdot q^r)$.
Hence
\begin{equation}\label{equation-g_d,q-main}
\func{d}{q}(a)-\func{d}{q}(a-1)=\sum_{i=0}^{i_0}\func{d-\sum_{j=i+1}^{r}a_j}{q}(a_i\cdot q^r)-\sum_{i=0}^{i_0}\func{d-\sum_{j=i+1}^{r}b_j}{q}(b_i\cdot q^r)\tag{*}
\end{equation}

We consider two cases based on the value of $i_0$. 

\bigskip\noindent
\underline{\textsc{Case  $i_0=0$:}}   In this case $b_{0}=a_{0}-1$,
and \eqref{equation-g_d,q-main} is simply
\[
\begin{split}
\func{d}{q}(a)-\func{d}{q}(a-1)
&=\func{d-\sum_{j=1}^r a_j}{q}(a_0)-\func{d-\sum_{j=1}^r a_j}{q}(a_0-1)\\
&=g_{d+1-\sum_{j=0}^r a_j}(1).
\end{split}
\]
Notice that $\func{d+1-\sum_{j=0}^r a_j}{q}(1)=\begin{cases}0 &\mbox{if }d+1<\sum_{j=0}^r a_j,\\ 1&\mbox{otherwise};\end{cases}$
so this is an increasing function of $d$.
Hence we have
$
\func{d}{q}(a)-\func{d}{q}(a-1)\geq \func{d'}{q}(a)-\func{d'}{q}(a-1)
$ 
as required.

\bigskip\noindent\underline{\textsc{Case $i_0>0$:}}
Then, it must be the case that $a_i=0, b_i=q-1$ for $0\leq i<i_0$.
Therefore, using the fact that $\func{d}{q}(0)=0$ for all $d$,
\eqref{equation-g_d,q-main} is
\[
\begin{split}
\func{d}{q}(a)-\func{d}{q}(a-1)
=&\func{d-\sum_{j=i_0+1}^r a_j}{q}(a_{i_0}\cdot q^{i_0})-\func{d-\sum_{j=i_0+1}^r b_j}{q}((a_{i_0}-1)\cdot q^{i_0})\\
&-\sum_{i=0}^{i_0-1}\func{d-\sum_{j=i+1}^{r}b_j}{q}((q-1)\cdot q^i)\\
=&\func{d+1-\sum_{j=i_0}^r a_j}{q}(q^{i_0})-\sum_{i=0}^{i_0-1}\func{d-\sum_{j=i+1}^r b_j}{q}((q-1)\cdot q^i)
\end{split}
\]
Define $h(i)=(\sum_{j=i_0}^r a_j)-1+(q-1)\cdot i$.
Then, we can verify that for $i\leq i_0-1$,
$$\sum_{j=i+1}^r b_j=(\sum_{j=i_0+1}^r a_j)+a_{i_0}-1+(q-1)(i_0-1-i)=h(i_0-1-i).$$ 

We claim that
$\displaystyle
\func{d-h(i)}{q}(q^{i_0-i})-\func{d-h(i)}{q}((q-1)\cdot q^{i_0-1-i})=\func{d-h(i+1)}{q}(q^{i_0-i-1}):
$

\medskip\noindent
Indeed, by Definition \ref{def-g_d,q},
\[
\func{d-h(i)}{q}((q-1)\cdot q^{i_0-i-1})=\sum_{j=0}^{p-2}\func{d-h(i)-j}{q}(q^{i_0-i-1}).
\]
By Proposition \ref{prop-expand},
we have
\[
\func{d-h(i)}{q}(q^{i_0-i})=\sum_{j=0}^{p-1}\func{d-h(i)-j}{q}(q^{i_0-i-1}).
\]
Comparing the two summations,
and using the fact $h(i)+q-1=h(i+1)$ we get the claim.

\medskip\noindent
So,
we can repeatedly apply the claim to get
\[
\begin{split}
\func{d}{q}(a)-\func{d}{q}(a-1)&=\func{d-h(0)}{q}(q^{i_0})-\sum_{i=0}^{i_0-1}\func{d-h(i)}{q}((q-1)\cdot q^{i_0-i-1})\\
&=\func{d-h(1)}{q}(q^{i_0-1})-\sum_{i=1}^{i_0-1}\func{d-h(i)}{q}((q-1)\cdot q^{i_0-i-1})\\
&=\cdots\\
&=\func{d-h(i_0-1)}{q}(q^{0})-\func{d-h(i_0-1)}{q}((q-1)\cdot q^{0})\\
&=\func{d+1-i_0\cdot (q-1)-\sum_{j=i_0}^r a_j}{q}(1).
\end{split}
\]
This is increasing in $d$,
so we obtain the required inequality.
\end{proof}

By telescoping Lemma~\ref{prop-main},
we have 
\begin{corollary}\label{prop-lifting}
For $a>b$, $d>d'$, for any $q>0$,
we have
\[
\func{d}{q}(a)-\func{d}{q}(b)\geq \func{d'}{q}(a)-\func{d'}{q}(b).
\]
\end{corollary}

Moreover,
if we set $b=0$,
then we can see that for fixed $a$,
$\func{d}{q}(a)$ is monotone in $d$.
\begin{corollary}\label{prop-monotone}
For $d>d'$, for any $q>0$,
we have
$
\func{d}{q}(a)\geq \func{d'}{q}(a).
$
\end{corollary}

\section{Proof of Lemma \ref{lem-main}}
\label{sec:sub-add}

In this section we prove the following generalized version of 
Lemma~\ref{lem-main} that applies to the generalization $g_{d,q}$ of $g_d$ given 
in the previous section.

\begin{lemma}\label{lem-main-stronger}
For any integer $q>0$,
for integers $a_0\geq a_1\geq \cdots \geq a_{q-1}\geq 0$, and
for any integer $d$,
\[
\func{d}{q}(\sum_{i=0}^{q-1}a_i)\leq \sum_{i=0}^{q-1}\func{d-i}{q}(a_i)
\]
\end{lemma}

\begin{proof}
Define $V\subset \mathbb{Z}^q$ as
$
V:=\{\vec{x}\in \mathbb{Z}^q\mid x_0\geq x_1\geq \cdots\geq x_{q-1}\geq 0\}.
$
For $\vec{a}\in V$,
define the \emph{value} of $\vec{a}$ of degree $d$ as
$
v_d(\vec{a}):=\sum_{i=0}^{q-1}\func{d-i}{q}(a_i).
$
With these notations,
Lemma \ref{lem-main-stronger} can be stated as $$\forall \vec{a}\in V,\quad
v_d(\vec{a})\geq \func{d}{q}(\|\vec{a}\|_1).$$

We define a total order $<$ on $V$ as $\vec{x}<\vec{y}$ iff either 
$\|\vec{x}\|_1<\|\vec{y}\|_1$,
or $\|\vec{x}\|_1=\|\vec{y}\|_1$ and $\vec{x}$ is lexicographically larger than 
$\vec{y}$.
Here we say that $\vec{x}$ is lexicographically larger than $\vec{y}$ if there 
exists and $i\geq 0$ so that $x_j=y_j$ for $j<i$,
and $x_j>y_j$.\footnote{Our definition of the order $<$ with lexicographically
larger is counterintuitive, but this is what we need.} 

We prove Lemma \ref{lem-main-stronger} by induction on this order $<$ over $V$.

\bigskip\noindent
\textbf{Induction Hypothesis:} For all $\vec{a'}\in V$ with $\vec{a'}<\vec{a}$,
\[
v_d(\vec{a'})\geq \func{d}{q}(\|\vec{a'}\|_1)\label{induction_hypothesis}\tag{**}
\]
In order to prove the inductive step we divide the proof into four cases depending
of the properties of $\vec{a}$:

We divide $V$ into 4 categories in which we prove the induction step $v_d(\vec{a})\geq \func{d}{q}(\|\vec{a}\|_1)$ with different methods.

Define $V_*:=\left\{\vec a\in V\mid  \exists t,r \text{ so that } a_0=\cdots=a_{t-1}=q^r,a_{t}<q^r,a_{t+1}=\cdots=a_{q-1}=0\right\}$.

\bigskip\noindent\textbf{CASE I $\quad\vec{a}\in V_{*}:\qquad$} 
The set $V_*$ is a special class of vectors for which we can directly show the $v_d(\vec{a})\geq \func{d}{q}(\|\vec{a}\|_1)$ without using the induction hypothesis.
Indeed,
for $\vec{a}\in V_*$,
by repeatedly applying Definition \ref{def-g_d,q}, we have
$$\func{d}{q}(\|\vec{a}\|_1)=\func{d}{q}(\sum_{i=0}^{q-1}a_i)=\func{d}{q}(t\cdot q^r+a_t)=\sum_{i=0}^{t-1}\func{d-i}{q}(q^r)+\func{d-t}{q}(a_t)=v_d(\vec{a})$$
So we are done if $\vec{a}\in V_*$.

\medskip\noindent
For the remaining cases,
we prove the inductive step by showing the following claim:

\begin{claim}\label{claim-main}
If $\vec{a}\not \in V_*$,
then $\exists \vec{a'}\in V$ so that $\vec{a'}<\vec{a}$,
$\|\vec{a'}\|_1=\|\vec{a}\|_1$ and $v_d(\vec{a'})\leq v_d(\vec{a})$.
\end{claim}
\noindent
Together with \eqref{induction_hypothesis},
we have
\[
\func{d}{q}(\|\vec{a}\|_1)=\func{d}{q}(\|\vec{a'}\|_1)\leq v_d(\vec{a'})\leq v_d(\vec{a})
\]
which is precisely what we need.
Our proof is algorithmic,
in the sense that we actually provide algorithms to construct $\vec{a'}$ explicitly.
Since we have different operations based on the structure of $\vec{a}$,
we introduce the structural properties that allow us to separate out the cases and
introduce the propositions that allow us to prove the claim in each case.
The proofs of these propositions are then completed in the following subsections.

Let $hp(\vec{a})$ (\emph{highest power} of $\vec{a}$) be the largest integer $r$ so that $a_0> q^r$.
Then for $hp(\vec{a})=r$ we can write $a_i=k_i\cdot q^r+c_i$, where $k_i\in \{0,\cdots,q\}$ and $0\leq c_i<q^r$.
We can divide $\vec{a}$ into groups so that in each group we have the same $k_i$.
That is,
we divide the interval $[0,q-1]$ into $\ell<p$ intervals $\{[s_i,t_i]\}_{i=1}^{\ell}$ where $s_i=0$,
$t_\ell=q-1$ and $s_{i+1}=t_i+1$, so that
\begin{itemize}
\item $\forall j\in[s_i,t_i]$, $k_j=k_{s_i}$, and
\item $k_{s_i}>k_{s_{i+1}}$ for $i=1,\cdots,\ell-1$.
\end{itemize}

\begin{defn}
We define $\groupheight{i}:=k_{s_i}$ as the \emph{height} of the interval $[s_i,t_i]$.
An interval $[s_i,t_i]$ is called \emph{singularized} if $c_{s_i+1}=\cdots=c_{t_i}=0$ when $h_i<q$,
and $c_{s_i}=\cdots=c_{t_i}=0$ when $h_i=q$.  (That is, at most the first element
in the interval has a non-zero remainder.)
$\vec{a}$ is \emph{singularized} if all its intervals are singularized. (Observe that every $\vec{a}\in V_{*}$ is singularized.)
\end{defn}

We show that if $\vec{a}$ is not singularized
then it can also be improved.

\bigskip\noindent\textbf{CASE II\quad} 
$\vec{a}$ is not singularized:\qquad In this case the inductive claim is an immediate consequence of the following proposition:

\begin{proposition}\label{prop-regroup}
Assuming induction hypothesis \eqref{induction_hypothesis},
if $\vec{a}$ is not singularized, then $\exists \vec{a'}\in V$ so that  $\vec{a'}<\vec{a}$,
$\|\vec{a'}\|_1=\|\vec{a}\|_1$ and $v_d(\vec{a'})\leq v_d(\vec{a})$.
\end{proposition}

It now suffices to consider singularized $\vec{a}$.

\begin{defn}\label{def-narrow}
Define the \emph{width} of $a_i$ to be $w_i:=\#\{i\leq j<q\mid a_j\geq q^r\}$.
An interval $[s_i,t_i]$ with height $h_i$ is called \emph{narrow} if $h_i\geq w_{s_i}$.
We say $\vec{a}$ is \emph{narrow} if all its intervals are narrow.  (Observe that
every $\vec{a}\in V_{*}$ is narrow.)
\end{defn}

If $\vec{a}$ singularized but not narrow, it can also be improved.

\bigskip\noindent\textbf{CASE III\quad} $\vec{a}$ is singularized but not narrow:\qquad
In this case the inductive claim is an immediate consequence of the following 
proposition:

\begin{proposition}\label{prop-transpose}
If $\vec{a}$ is singularized but not narrow,
then $\exists \vec{a'}\in V$ so that  $\vec{a'}<\vec{a}$,
$\|\vec{a'}\|_1=\|\vec{a}\|_1$ and $v_d(\vec{a'})=v_d(\vec{a})$.
\end{proposition}

Finally, we consider the remaining case that $\vec{a}$ is singularized, narrow, and
not in $V_*$.

\bigskip\noindent\textbf{CASE IV\quad} $\vec{a}\not \in V_{*}$,
but $\vec{a}$ is both singularized  and narrow:\qquad
In this case the inductive claim is an immediate consequence of the following
proposition:

\begin{proposition}\label{prop-repack}
Assuming induction hypothesis \eqref{induction_hypothesis},
for $\vec{a} \in V \slash V_{*}$,
if $\vec{a}$ is singularized and narrow,
then $\exists \vec{a'}\in V$ so t{}hat  $\vec{a'}<\vec{a}$,
$\|\vec{a'}\|_1=\|\vec{a}\|_1$ and $v_d(\vec{a'})\leq v_d(\vec{a})$.
\end{proposition}

\noindent
Now, assuming Propositions \ref{prop-regroup}, \ref{prop-transpose} and \ref{prop-repack}, which we prove in the following sections, and putting Cases II, III, and IV together, we immediately derive
Claim \ref{claim-main}.

From Case I if $\vec{a}\in V_{*}$,
then $v_d(\vec{a})=\func{d}{q}(\|\vec{a}\|_1)$ which is sufficient.
Otherwise, we apply Claim \ref{claim-main} to obtain $\vec{a'}\in V$ so that $\vec{a'}<\vec{a}$,
$\|\vec{a'}\|_1=\|\vec{a}\|_1$ and $v_d(\vec{a'})\leq v_d(\vec{a})$.
By the induction hypothesis \eqref{induction_hypothesis},
$v_d(\vec{a'})\geq \func{d}{q}(\|\vec{a'}\|_1)$.
Therefore
\[
\func{d}{q}(\|\vec{a}\|_1)=\func{d}{q}(\|\vec{a'}\|_1)\leq v_d(\vec{a'})\leq v_d(\vec{a})
\]
as required and the lemma follows by induction.
\end{proof}

\noindent
We now finish the overall argument by proving each of the Propositions \ref{prop-regroup}, \ref{prop-transpose} and \ref{prop-repack} in each of the following subsections.

\subsection{Singularization}

In this subsection we prove Proposition \ref{prop-regroup}.
\begin{proof}[Proof of Proposition \ref{prop-regroup}]
We claim that the following algorithm on input $\vec{a}$ that is not singularized,
outputs $\vec{a'}\in V$ so that  $\vec{a'}<\vec{a}$,
$\|\vec{a'}\|_1=\|\vec{a}\|_1$ and $v_d(\vec{a'})\leq v_d(\vec{a})$.
\bigskip\\

\begin{algorithm}[H]\label{algorithm-regroup}
\caption{Singularization}
\KwIn {$\vec{a}=(a_0,\cdots,a_{q-1})$.}
\KwOut {Singularized $\vec{a'}$.}
\For{$i=1,\cdots,\ell$}{
  Compute $d_i=\sum_{j=s_i}^{t_i} c_j$.\\
  Compute $0\leq e_i<q,0\leq f_i<q^t$ so that $d_i=e_i\cdot q^r+f_i$.\\
  \tcc{Here we use the fact that $e_i<\frac 1 {q^r}\cdot (t_i-s_i+1)\cdot q^r=t_i-s_i+1$}
  $c_{s_i+j}'\leftarrow \begin{cases} f_i & j=e_i\\ 0 & \mbox{else}\end{cases}$.\\
  $k_j'\leftarrow \begin{cases} k_{j}+1 & j=s_i,\cdots,s_{i}+e_i-1\\ 0 & j=s_i+e_i,\cdots,t_i\end{cases}$.\\
}
\For{$i=0,\cdots,q-1$}{
  $a_i'=k_i'\cdot q^r+c_i'$
}
\end{algorithm}

\begin{figure} 
\centering
    \includegraphics[width=0.5\textwidth]{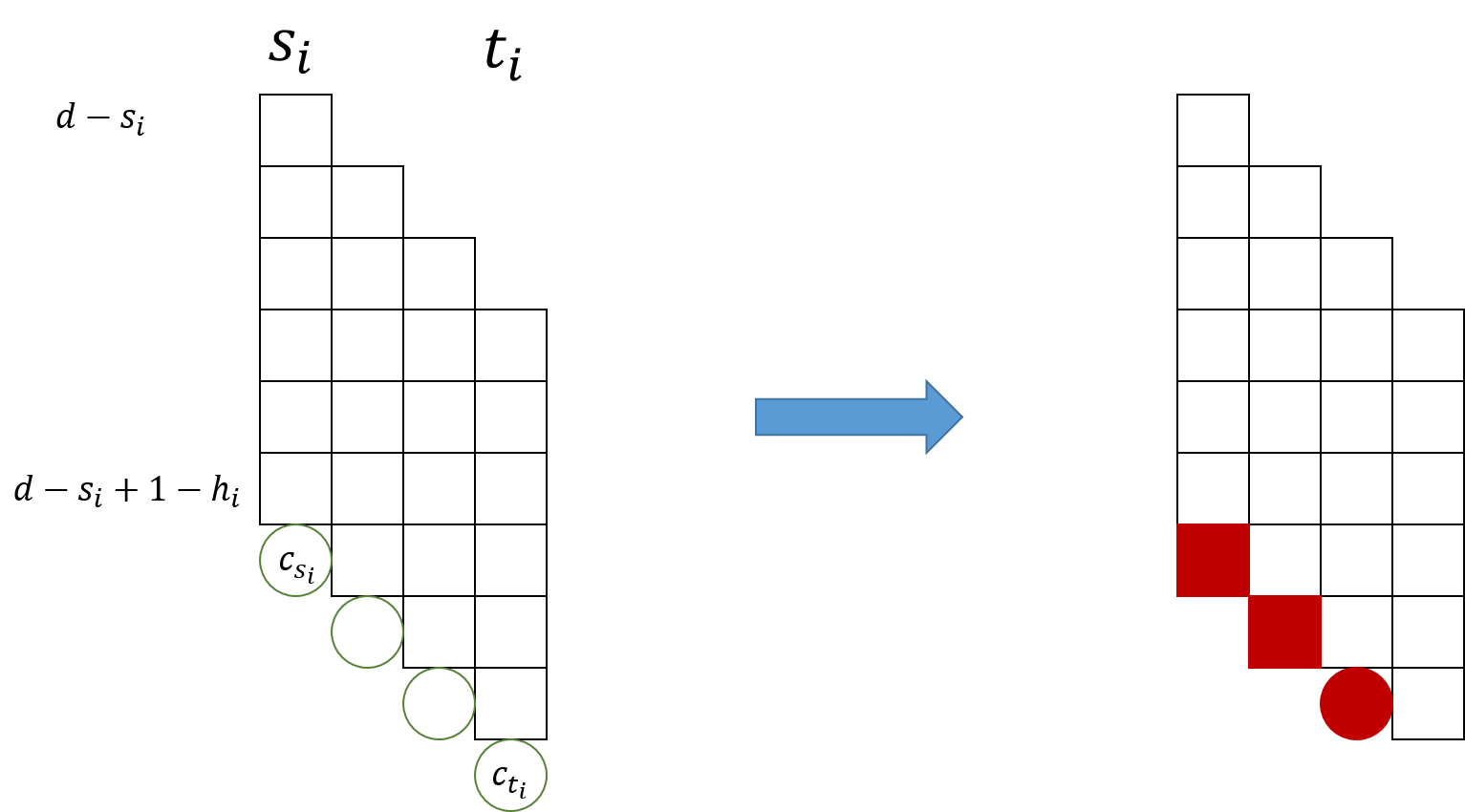}
  \caption{A visualization of Algorithm \ref{algorithm-regroup} on the interval $[s_i,t_i]$. In the diagram, columns are indexed by $\{0,\cdots,q-1\}$ and rows are index by integers in decreasing order. A box on level $d$ represents $\func{d}{q}(q^r)$, and a circle on level $d$ labeled by $c$ represents $\func{d}{q}(c)$. The diagram represents $v_d(\vec{a})$. Since $\func{d-i}{q}(a_i)=\func{d-i}{q}(k_i\cdot q^r+c_i)=\sum_{j=0}^{k_i-1}\func{d-i-j}{q}(q^r)+\func{d-i-k_i}{q}(c_i)$,
  in column $j$ we have 1 box from level $d-i$ to $d-i+1-k_i$, and 1 circle labeled by $c_{i}$ on level $d-i-k_i$.
Algorithm \ref{algorithm-regroup} turns the original diagram on the left to the new diagram on the right. In the right diagram we omit all the circles labeled with 0.
  } 
\end{figure}

\bigskip\noindent
First,
inside each original group $[s_i,t_i]$,
we always have $a_{j}'\geq a_{j+1}'$ by construction.
Moreover,
between groups we have
\[
a_{t_i}'=\groupheight{i}\cdot q^r>\groupheight{i+1}\cdot q^r\geq (\groupheight{i+1}+1)\cdot q^r\geq a_{s_{i+1}}'
\] 
and therefore $\vec{a'}\in V$.
Now let us consider $\func{d-j}{q}(a_j)-\func{d-j}{q}(a_j')$ for $j\in [s_i,t_i]$.
By the algorithm and definition \ref{def-g_d,q},
we have
\[
\begin{split}
\func{d-j}{q}(a_j)-\func{d-j}{q}(a_j')=&\sum_{m=0}^{\groupheight{i}-1}\func{d-j-m}{q}(q^r)+\func{d-j-\groupheight{i}}{q}(c_j)\\
&-[\sum_{m=0}^{\groupheight{i}-1}\func{d-j-m}{q}(q^r)+\func{d-j-\groupheight{i}}{q}(a_j'-\groupheight{i}\cdot q^r)]\\
=&\func{d-j-\groupheight{i}}{q}(c_j)-\func{d-j-\groupheight{i}}{q}(a_j'-\groupheight{i}\cdot q^r);
\end{split}
\]
so,
\[
\begin{split}
\sum_{j=s_i}^{t_i} \func{d-j}{q}(a_j)-\sum_{j=s_i}^{t_i} \func{d-j}{q}(a_j')&=\sum_{j=s_i}^{t_i} \func{d-j-\groupheight{i}}{q}(c_j)-\sum_{j=s_i}^{t_i} \func{d-j-\groupheight{i}}{q}(a_j'-q^r)\\
&=\sum_{j=s_i}^{t_i}\func{d-j-\groupheight{i}}{q}(c_j)-[\sum_{j=s_i}^{s_i+e_i-1}\func{d-j-\groupheight{i}}{q}(q^r)+\func{d-s_i-e_i-\groupheight{i}}{q}(f_i)]\\
&=\sum_{j=s_i}^{t_i}\func{d-j-\groupheight{i}}{q}(c_j)-\func{d-s_i-\groupheight{i}}{q}(e_i\cdot q^r+f_i).
\end{split}
\]

This is the place where we can use the induction hypothesis.
Notice that $t_i-s_i+1\leq q$.
Consider the configuration $\vec{c}=(c_{s_i},c_{s_i+1},\cdots,c_{t_i},0,\cdots,0)$.
Then $\|\vec{c}\|_1<\sum_{j=s_i}^{t_i}a_i\leq \|\vec{a}\|_1$.
Moreover,
$c_{s_i}\geq c_{s_i+1}\geq \cdots \geq c_{t_i}\geq 0$.
So by \eqref{induction_hypothesis} we have
\[
\sum_{j=s_i}^{t_i}\func{d-j-\groupheight{i}}{q}(c_j)=v_{d-s_i-\groupheight{i}}(\vec{c})\geq \func{d-s_i-\groupheight{i}}{q}(\sum_{j=s_i}^{t_i}c_i)=\func{d-s_i-\groupheight{i}}{q}(e_i\cdot q^r+f_i)
\]
which gives
\[
\sum_{j=s_i}^{t_i} \func{d-j}{q}(a_j)-\sum_{j=s_i}^{t_i} \func{d-j}{q}(a_j')\geq 0.
\]
This holds for every original interval. So in total we have $v_d(\vec{a})\geq v_d(\vec{a'})$.
Also,
the construction clearly gives us
\[
\|\vec{a}\|_1=\sum_{i=1}^l\sum_{j=s_i}^{t_i} a_j=\sum_{i=1}^l\sum_{j=s_i}^{t_i} (\groupheight{i}\cdot q^r+c_j)=\sum_{i=1}^l((\sum_{j=s_i}^{t_i} \groupheight{i}\cdot q^r)+e_i\cdot q^r+f_i)=\sum_{i=1}^l\sum_{j=s_i}^{t_i} a_j'=\|\vec{a'}\|_1
\]

Let $[s_i,t_i]$ be the first original interval that is not singularized. 
Then by construction we can see that $a_j=a_j'$ for $j<s_i$.
If $k_{s_i}'>k_{s_i}$,
then clearly $a_{s_i}<a_{s_i}'$.
Otherwise $c_{s_i}'=\sum_{j=s_i}^{t_i} c_{j}>c_{s_i}$ because $[s_i,t_i]$ is not singularized.
Again we have $a_{s_i}<a_{s_i}'$.
So $\vec{a}$ is strictly lexicographically larger than $\vec{a'}$,
which implies  $\vec{a'}<\vec{a}$.
\end{proof}

\subsection{Transposing}

In this subsection we prove Proposition \ref{prop-transpose}.
For convenience of notations,
we use $w_{i,j}:=\#\{i\leq t<q\mid a_t\geq j\cdot q^r\}$ which we call the \emph{$j$-th order width} of $a_i$.

\begin{proof}[Proof of Proposition \ref{prop-transpose}]
We claim that the following algorithm on input $\vec{a}$ that is not narrow,
outputs  $\vec{a'}\in V$ so that  $\vec{a'}<\vec{a}$,
$\|\vec{a'}\|_1=\|\vec{a}\|_1$ and $v_d(\vec{a'})=v_d(\vec{a})$.

\bigskip
\begin{algorithm}[H]\label{algorithm-transpose}
\caption{Transposing}
\KwIn {Singularized $\vec{a}=(a_0,\cdots,a_{q-1})$.}
\KwOut {$\vec{a'}$.}
Let $i_0$ be the smallest integer so that $[s_{i_0},t_{i_0}]$ is not narrow.\\
\For{$i=0,\cdots,q-1$}{
  \If{$i<s_{i_0}$}{
  \tcc{We do not change numbers before $i_0$-th group}
  $a_i'\leftarrow a_i$
  }
  \Else{
  \tcc{first set $k_i'$}
  $k_i'\leftarrow w_{s_{i_0},1+i-s_{i_0}}$ \\
  $c_i'\leftarrow 0$
  }
}
\For{$i=i_0,\cdots,\ell$}{
  $c_{s_{i_0}+h_{i}}'\leftarrow c_{s_{i}}$
}

\For{$i=0,\cdots,q-1$}{
  $a_i'=k_i'\cdot q^r+c_i'$
}
\end{algorithm}

\begin{figure} 
\centering
    \includegraphics[width=0.5\textwidth]{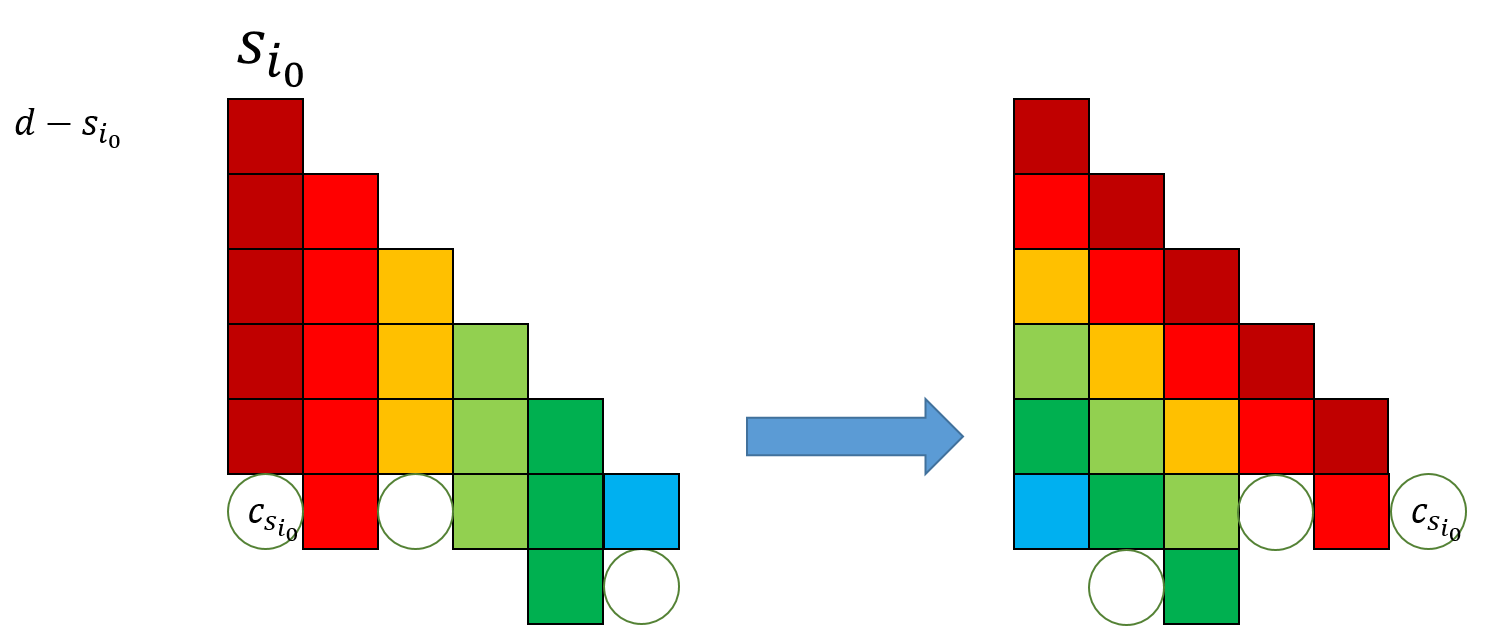}
  \caption{A visualization of Algorithm \ref{algorithm-transpose}.
  } 
\end{figure}

\bigskip\noindent
We first argue that $\vec{a'}\in V$.
Notice that $w_{s_{i_0},1+i-s_{i_0}}$ is a decreasing function on $i$.
By construction,
for $i\geq s_{i_0'}$,
$k_{i}'\geq k_{i+1}'$.
When $k_i'>k_{i+1}'$,
clearly $a_i'\geq k_i'\cdot q^{r}\geq a_{i+1}'$.
When $k_i'=k_{i+1}'$,
we claim that $c_{i+1}'$ must be 0.
This is because for contradiction assume $i+1=s_{i_0}+h_q$ for some $q\in \{i_0,\cdots,\ell\}$,
then by setting $j=s_{q}$,
we have $k_j=k_{s_q}=h_q=i+1-s_{i_0}$,
while $k_j<(i+1)+1-s_{i_0}$,
which implies $k_i'>k_{i+1}'$.
So $a_{i}'\geq k_i'\cdot q^r=a_{i+1}'$.
We conclude that when $i\geq i_0$,
$a_i'\geq a_{i+1}'$.
If $i_0>1$,
because $\vec{a}\in V$,
we have $a_i'=a_i\geq a_{i+1}=a_{i+1}'$ when $i+1<s_{i_0}$.
So we only need to show that $a_{s_{i_0}-1}'\geq a_{s_{i_0}}'$.
To see this,
because $[s_{i_0-1},t_{i_0-1}]$ is narrow,
we have $h_{i_0-1}\geq w_{s_{i_0-1},1}$.
Also,
$h_{i_0-1}\geq 1$ otherwise we will not have the next group.
So
\[
a_{s_{i_0}-1}'=a_{s_{i_0}-1}=a_{t_{i_0-1}}\geq h_{i_0-1}\cdot q^r\geq (w_{s_{i_0-1},1})\cdot q^r\geq (1+w_{s_{i_0},1})\cdot q^r=(1+k_{s_{i_0}}')\geq a_{s_{i_0}}'.
\]
Therefore $\vec{a'}\in V$.  Then we argue that $\|\vec{a'}\|_1=\|\vec{a}\|_1$.
This is because for singularized $\vec{a}$,
\[
\begin{split}
\|\vec{a}\|_1=\sum_{j=0}^{q-1}a_j=\sum_{j=0}^{s_{i_0}-1}a_j+(\sum_{j=s_{i_0}}^{q-1} k_j)\cdot q^r +\sum_{i=i_0}^{\ell} c_{s_{\ell}}.
\end{split}
\]
On the other hand,
\[
\begin{split}
\|\vec{a}'\|_1=\sum_{j=0}^{q-1}a_j'&=\sum_{j=0}^{s_{i_0}-1}a_j'+(\sum_{j=s_{i_0}}^{q-1} k_j')\cdot q^r +\sum_{i=i_0}^{\ell} c_{s_{i_0}+h_i}'\\
&=\sum_{j=0}^{s_{i_0}-1}a_j+(\sum_{j=s_{i_0}}^{q-1} k_j')\cdot q^r +\sum_{i=i_0}^{\ell} c_{s_{i_0}+h_i}'.
\end{split}
\]
But we have
\[
\sum_{i=i_0}^{\ell} c_{s_{i_0}+h_i}'=\sum_{i=i_0}^{\ell} c_{s_{\ell}}.
\]
Notice that when $i\geq s_{i_0}$, we have $k_i\leq \hone$,
so 
\[
\sum_{j=s_{i_0}}^{q-1} k_j=\sum_{j=s_{i_0}}^{q-1}(\sum_{i=1}^{\hone}\mathbf{1}_{k_j\geq i})=\sum_{i=1}^{\hone}\sum_{j=s_{i_0}}^{q-1}\mathbf{1}_{k_j\geq i}
\]
and
\[
\sum_{j=s_{i_0}}^{q-1} k_j'=\sum_{i=s_{i_0}}^{q-1}w_{s_{i_0},i+1-s_{i_0}}=\sum_{i=s_{i_0}}^{q-1}(\sum_{j=s_{i_0}}^{q-1} \mathbf{1}_{k_j\geq (i+1-s_{i_0})})=\sum_{i=1}^{p-s_{i_0}}\sum_{j=s_{i_0}}^{q-1}\mathbf{1}_{k_j\geq i}=\sum_{i=1}^{\hone}\sum_{j=s_{i_0}}^{q-1}\mathbf{1}_{k_j\geq i}.
\]
Here we use the fact that $[s_{i_0},t_{i_0}]$ is not narrow,
hence $\hone\leq q-s_{i_0}$.
Combining these lines gives us $\|\vec{a'}\|_1=\|\vec{a}\|_1$.

We next argue that  $\vec{a'}$ is lexicographically larger than $\vec{a}$:
For $i<s_{i_0}$,
we have $a_i'=a_i$.
But since $[s_{i_0},t_{i_0}]$ is not narrow, we have
\[
a_{s_{i_0}}'=(\sum_{j=s_{i_0}}^{q-1}\mathbf{1}_{k_j\geq 1})\cdot q^r \geq (h_{s_{i_0}}+1)\cdot q^r\geq a_{s_{i_0}}.
\]
Therefore $\vec{a'}<\vec{a}$.   We finally show that $v_d(\vec{a'})=v_d(\vec{a})$:
We have 
\[
\begin{split}
v_d(\vec{a})&=\sum_{i=0}^{q-1}g_{d-i}(a_i)=\sum_{i=0}^{s_{i_0}-1}g_{d-i}(a_i)+\sum_{i=s_{i_0}}^{q-1}g_{d-i}(k_i\cdot q^r+c_i)\\
&=\sum_{i=0}^{s_{i_0}-1}g_{d-i}(a_i)+(\sum_{i=s_{i_0}}^{q-1}\sum_{j=0}^{k_i-1}g_{d-i-j}(q^r))+\sum_{i=i_0}^{\ell}g_{d-s_i-h_i}(c_{s_i})
\end{split}
\]
On the other hand,
we have for $i=i_0,\cdots,\ell$,
$k_{s_{i_0}+h_i}'=\sum_{j=s_{i_0}}^{q-1}\mathbf{1}_{k_j\geq h_i+1}=s_i-s_{i_0}$.
So
\[
v_d(\vec{a'})=\sum_{i=0}^{s_{i_0}-1}g_{d-i}(a_i')+(\sum_{i=s_{i_0}}^{q-1}g_{d-i}(k_i'\cdot q^r))+\sum_{i=i_0}^{\ell}g_{d-(s_{i_0}+h_i)-(s_i-s_{i_0})}(c_{s_i})
\]
Comparing the two expressions,
it is sufficient to show that $$\sum_{i=s_{i_0}}^{q-1}\sum_{j=0}^{k_i-1}g_{d-i-j}(q^r)=\sum_{i=s_{i_0}}^{q-1}g_{d-i}(k_i'\cdot q^r).$$
Indeed,
we have
\[
\begin{split}
\sum_{i=s_{i_0}}^{q-1}\sum_{j=0}^{k_i-1}g_{d-i-j}(q^r)
&=\sum_{j=0}^{\hone-1}\sum_{i=s_{i_0}}^{q-1}\mathbf{1}_{k_i\geq j+1}\cdot g_{d-i-j}(q^r)\\
&=\sum_{j=0}^{\hone-1}g_{\done-j}((\sum_{i=s_{i_0}}^{q-1}\mathbf{1}_{k_i\geq j+1})\cdot q^r)\\
&=\sum_{j=0}^{\hone-1}g_{\done-j}(k_{s_0+j}'\cdot q^r)\\
&=\sum_{i=s_{i_0}}^{q-1}g_{d-i}(k_i'\cdot q^r)
\end{split}
\]
where we use the fact that $k_i'=0$ when $i\geq s_{i_0}+\hone$.
\end{proof}

\subsection{Repacking}

In this section we prove Proposition \ref{prop-repack}.
We need one auxiliary lemma.

\begin{lemma}\label{lem-subadd}
Assuming induction hypothesis \eqref{induction_hypothesis},
then for all integers $ x,y\geq 0$ so that $x+y< \|\vec{a}\|_1$, for every integer $d\geq 0$,
$$\func{d}{q}(x)+\func{d}{q}(y)\geq \func{d}{q}(x+y).$$
\end{lemma}

\begin{proof}
Without loss of generality assume that $x\geq y$.
Then consider $\vec{b}=(x,y,0,\cdots,0)$.
We can verify that $\vec{b}\in V$ and $\vec{b}<\vec{a}$.
Hence $\func{d}{q}(x)+\func{d-1}{q}(y)=v_d(\vec{b})\geq \func{d}{q}(\|\vec{b}\|_1)=\func{d}{q}(x+y)$.
Now,
by Corollary \ref{prop-monotone},
$\func{d-1}{q}(y)\leq \func{d}{q}{y}$,
therefore 
\[\func{d}{q}(x)+\func{d}{q}(y)\geq \func{d}{q}(x)+\func{d-1}{q}(y)\geq \func{d}{q}(x+y)\\[-4ex]\]
\end{proof}

Now consider $\vec{a}\not \in V_{*}$ that is singularized and narrow.
Vectors in $V_{*}$ are very structured in the sense that they can have at most 3 heights.
Inspired by the definition of $V_{*}$,
we have the following definition.

\begin{defn}\label{def-packed}
An interval $[s,q-1]$ is called \emph{$k$-packed} if one of the following happens:
\begin{itemize}
\item $a_s=k\cdot q^r+c_s$ and $a_{s+1}=\cdots=a_{q-1}=k\cdot q^r$ where $0\leq c_s<q^r$ or 
\item $\exists t> s$ so that $a_s=k\cdot q^r+c_s$, $a_{s+1}=\cdots=a_{t-1}=k\cdot q^r$, $a_t<k\cdot q^r$ and $a_{t+1}=\cdots =a_{q-1}=0$.
\end{itemize}
Interval $[s,q-1]$ is called packed if there exists a $k$ so that $[s,q-1]$ is $k$-packed.
\end{defn}

With this definition,
we can verify that $\forall \vec{b} \in V_{*}$,
$\vec{b}$ is $q$-packed on $[0,q-1]$.
Since $\vec{a}\not \in V_{*}$,
let $i$ be the smallest integer so that $[i,q-1]$ is packed;
then either $i\neq 0$,
or $[0,q-1]$ is $k$-packed for some $k<q$.
We are going to rearrange $[i,q-1]$ so that it is $(k+1)$-packed to improve $\vec{a}$.

\begin{proof}[Proof of Proposition \ref{prop-repack}]

We claim that the following operation on input $\vec{a}\not\in V_{*}$ that is singularized and narrow,
outputs  $\vec{a'}\in V$ so that  $\vec{a'}<\vec{a}$,
$\|\vec{a'}\|_1=\|\vec{a}\|_1$ and $v_d(\vec{a'})\leq v_d(\vec{a})$.

We first select the smallest $i_0$ so that $[s_{i_0},q-1]$ is packed.
Then $[s_{i_0},q-1]$ is $k$-packed where $k$ is the height of the interval $[s_{i_0},t_{i_0}]$.
Then we set $a_i'=a_i$ for all $i<s_{i_0}$,
that is,
we only ``repack'' the interval $[s_{i_0},q-1]$.
We compute $\csum=\sum_{j=s_{i_0}}^{q-1} a_i$,
then compute $t=\lfloor \frac{\csum}{(h_{i_0}+1)\cdot q^r} \rfloor$ and $c=\csum-t\cdot (h_{i_0}+1)\cdot q^r$.
Then we set $a_{s_{i_0}+j}'=(h_{i_0}+1)\cdot q^r$ for $j=0,\cdots,t-1$;
set $a_{s_{i_0}+t}=c$;
and all the other $a_i'$ to be 0.

\begin{figure} 
\centering
    \includegraphics[width=0.5\textwidth]{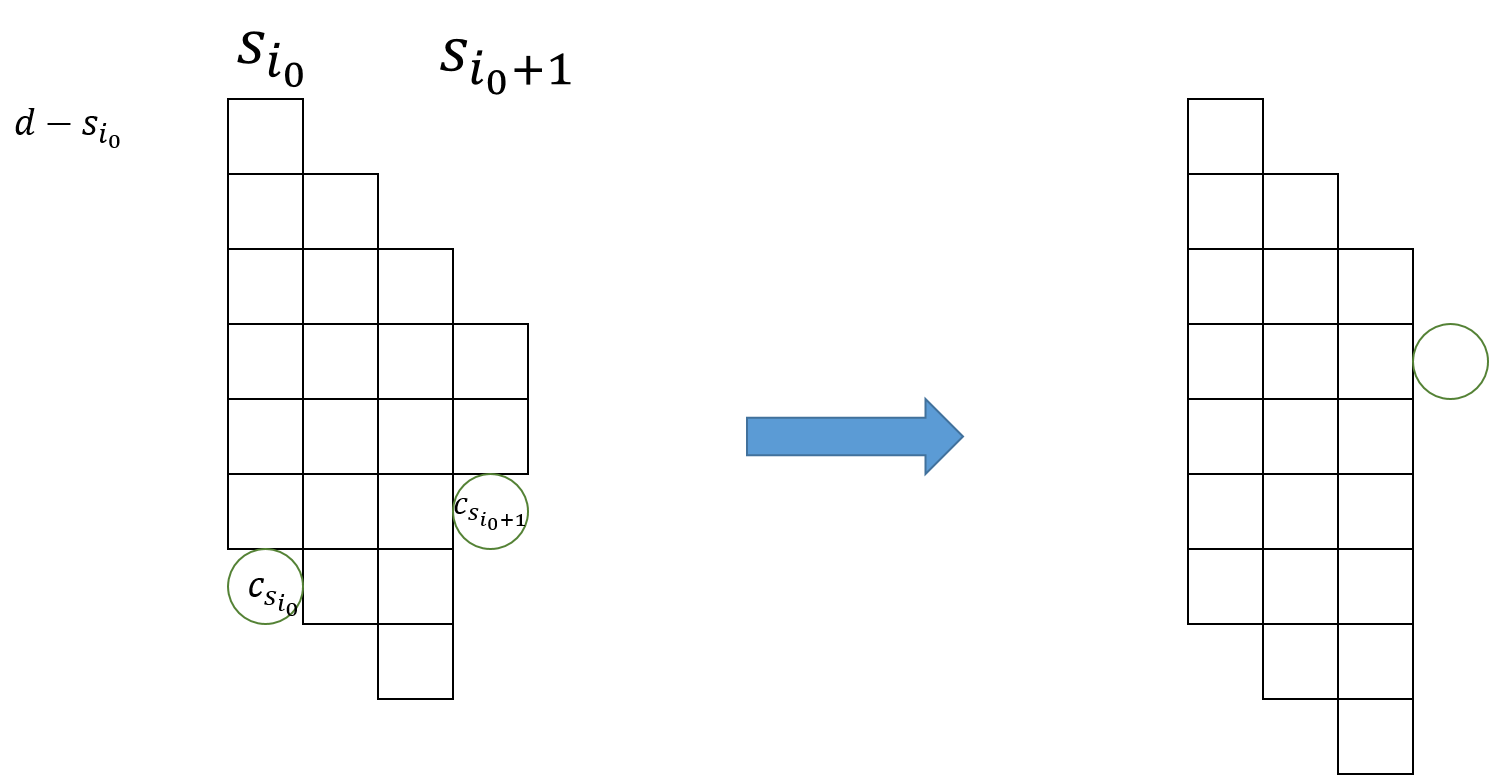}
  \caption{A visualization of the operation.
  } 
\end{figure}

We first argue that $h_{i_0}>0$.
Here we use the fact that $\vec{a}$ is singularized,
therefore when $h_{i_0}=0$,
$a_j=0$ for all $j>s_{i_0}$.
If $i_0=1$,
then
we have $\vec{a}\in V_{*}$ because when $h_{1}=0$,
 $\vec{a}$ is of the form $(a_0,0,\cdots,0)$;
if $i_0>1$,
we can see that $[s_{i_0-1},q-1]$ is also packed,
which contradicts the assumption of $i_0$.

Now we are ready to prove properties of $\vec{a'}$.
First we show that $\vec{a'}\in V$.
This is obvious if $i_0=1$.
Otherwise,
by construction we only need to show that $a_{s_{i_0}-1}'\geq a_{s_{i_0}}'$.
Since $h_{i_0}+1\leq h_{i_0-1}$,
we have $a_{s_{i_0}-1}'\geq h_{i_0-1}\cdot q^r\geq (h_{i_0}+1)\cdot q^r=a_{s_{i_0}}'$.
So $\vec{a'}\in V$.

By construction,
we can see
$\|\vec{a}\|_1=\|\vec{a'}\|_1$.
Now we argue that $\vec{a'}<\vec{a}$.
If $t\neq 0$,
this is obvious from construction.
Otherwise if $t=0$,
first we must have $s_{i_0}=t_{i_0}$,
otherwise since $h_{i_0}>0$,
$\csum=\sum_{j=s_{i_0}}^{q-1} a_i>a_{s_{i_0}}+a_{t_{i_0}}\geq 2h_{i_0}\cdot q^r\geq (h_{i_0}+1)\cdot q^r$.
But by the choice of $i_0$,
we must have $t_{i_0}<q-1$ and $a_{t_{i_0}+1}>0$,
otherwise we could choose $i_0-1$.
So $a_{s_{i_0}}'=\sum_{j=s_{i_0}}^{q-1}a_j\geq a_{s_{i_0}}+a_{t_{i_0}+1}>a_{s_{i_0}}$,
which implies that $\vec{a'}< \vec{a}$.

The proposition is immediate from the following claim whose proof
is somewhat tedious.

\begin{claim}\label{prop-repacking-long-proof}
Let $\vec{a'}$ be the vector obtained by repacking.
Assuming induction hypothesis \eqref{induction_hypothesis},
we have
$
v_d(\vec{a})\geq v_d(\vec{a'})$.
\end{claim}

It remains to prove this claim.
By construction,
when $i<s_{i_0}$,
$a_i'=a_i$,
hence 
\begin{align}\label{eq-diff}
v_d(\vec{a'})-v_d(\vec{a})=\sum_{j=s_{i_0}}^{q-1} \func{d-j}{q}(a_j')-\sum_{j=s_{i_0}}^{q-1} \func{d-j}{q}(a_j)
\end{align}

To prove the claim,
we will expand the $\func{d}{q}$ function in the summations.
We will see that those summations have many terms in common,
so we can do a lot of cancellation.
Furthermore,
we will use Corollary \ref{prop-lifting},
which allows us to simplify the expression greatly.
Then we are able to show that after simplification,
the right-hand side of \eqref{eq-diff} is no more than 0 and hence the
claim follows.

Since $\vec{a}$ is singularized and $[s_{i_0},q-1]$ is packed,
without losing generality we may assume $\vec{a}$ is of the form
$$a_i=\begin{cases} h_{i_0}\cdot q^r+c_{s_{i_0}}&i=s_{i_0}\\  
 h_{i_0}\cdot q^r& s_{i_0}<i\leq t_{i_0}\\  
 h_{i_0+1}\cdot q^r+c_{s_{i_0}+1}& \mbox{if } i=t_{i_0}+1\\  
0&\mbox{if }  i>t_{i_0}+1.\\  
\end{cases}$$
Here, $h_{i_0}$ is the height of the interval $[s_{i_0},t_{i_0}]$
and $h_{i_0+1}$ is the height of the interval $[s_{i_0+1},t_{i_0+1}]$.
To avoid unnecessary special cases,
we set $h_{i_0+1}=c_{s_{i_0}+1}=0$ if $i_0=\ell$,
then we can view all four of these cases as occurring.

Notice that in \eqref{eq-diff},
all the terms in the summations start with $j=s_{i_0}$.
To simplify notation, we shift the indices so that we start from $0$.
For $j\geq 0$,
we set $b_j=a_{j+s_{i_0}}$ and $b_j'=a_{j+s_{i_0}}'$.
Then $\vec{b}$ is a vector of length $q-s_{i_0}$. 
Let $d'=d-s_{i_0}$ and $t'=t_{i_0}-s_{i_0}+1$.
Also in order to drop messy subscripts,
we set
$c_1=c_{s_{i_0}}$,
$c_2=c_{s_{i_0}+1}$,
$k=h_{i_0}$ 
and $k'=h_{i_0+1}$.

By replacing symbols,
we have
$$b_i=\begin{cases} k\cdot q^r+\cone &i=0\\  
k\cdot q^r& 0<i< t'\\  
 k'\cdot q^r+\ctwo& \mbox{if } i=t'\\  
0&\mbox{if }  i>t'. 
\end{cases}$$
Hence we can compute $\csum=\sum_{i=0}^{t'} b_i=\tlength\cdot \hone\cdot q^r+\htwo\cdot q^r+\cone+\ctwo$.
Also, since
$\vec{a}$ is narrow,
 by Definition \ref{def-narrow},
 $\hone\geq \tlength$.
Recall that $t=\lfloor \frac{\csum}{(\hone+1)\cdot q^r}\rfloor$.
We argue that $\tlength-1\leq t\leq \tlength$.
Indeed,
since
\[
\tlength\cdot \hone-(\tlength-1)(\hone+1)=\hone-(\tlength-1)\geq 0
\]
we get $\tlength-1\leq t$.
On the other hand,
if $t>\tlength$,
then $\csum\geq (\tlength+1)\cdot (\hone+1)\cdot q^r$,
which is equivalent to
\[
\htwo\cdot q^r+\cone+\ctwo>(\tlength+\hone+1)\cdot q^r.
\]
But this cannot be true,
since $\htwo<\hone$ and $\cone+\ctwo<2\cdot q^r$.

\medskip\noindent
So we know that $t$ can only be $\tlength-1$ or $\tlength$.
Observe that
\begin{align}\label{eq-expand-a'}
\sum_{j=s_{i_0}}^{q-1} \func{d-j}{q}(a_j')=\sum_{j=0}^{q-1-s_{i_0}} \func{d'}{q}(b_j')=\sum_{j=0}^{t-1}\func{\done-j}{q}((\hone+1)\cdot q^r)+\func{\done-t}{q}(c)
\end{align}
and
\begin{align}\label{eq-expand-a}
\sum_{j=s_{i_0}}^{q-1} \func{d-j}{q}(a_j)=\sum_{j=0}^{q-1-s_{i_0}} \func{d'}{q}(b_j)&=\func{\done}{q}(\hone\cdot q^r+\cone)+\sum_{j=1}^{\tlength-1}\func{\done-j}{q}(\hone\cdot q^r)+\func{\dtwon}{q}(\htwo\cdot q^r+\ctwo)\nonumber \\
&=\sum_{j=0}^{\tlength-1}\func{\done-j}{q}(\hone\cdot q^r)+\func{\done-\hone}{q}(\cone)+\func{\dtwon}{q}(\htwo\cdot q^r+\ctwo).
\end{align}

These two summations are very similar in the sense that we can break them into two parts:
one structured summation where the input of the $\func{d}{q}$ function is some power of $q^r$,
and another part where the input is quite irregular.
If $t=t'$,
then we can ``align'' the structured part;
otherwise, there will be one more term in \eqref{eq-expand-a}.
We break things into cases based on the value of $t$.

\bigskip\noindent\underline{\textsc{Case 1}}\quad $t=\tlength-1$:\quad If this happens,
we must have $\csum<\tlength\cdot (\hone+1)\cdot q^r$,
which is equivalent to
\[
\htwo\cdot q^r+\cone+\ctwo<\tlength\cdot q^r
\]
Notice that for $j=0,\cdots,t-1$, by Definition \ref{def-g_d,q} we have
\[
\func{\done-j}{q}((\hone+1)\cdot q^r)-\func{\done-j}{q}(\hone\cdot q^r)=\func{\done-\hone-j}{q}(q^r).
\]
Together with \eqref{eq-expand-a}  and \eqref{eq-expand-a'}, \eqref{eq-diff} becomes
\[
\begin{split}
v_d(\vec{a'})-v_d(\vec{a})=&\sum_{j=0}^{t-1}\func{\done-\hone-j}{q}(q^r)+\func{\done-t}{q}(c)\\&-(\func{d'-t'+1}{q}(\hone\cdot q^r)+\func{\done-\hone}{q}(\cone)+\func{\dtwon}{q}(\htwo\cdot q^r+\ctwo)).
\end{split}
\]
Since we have 
\[
c=\csum-t\cdot (\hone+1)\cdot q^r=(\hone-(\tlength-1))\cdot q^r+\htwo\cdot q^r+\cone+\ctwo
\]
and since $c<q^{r+1}$ and $\hone-(\tlength-1)\geq 1$,
by Definition \ref{def-g_d,q} we have
\[
\func{\done-t}{q}(c)=\sum_{j=0}^{\hone-\tlength}\func{d'-t'-j}{q}(q^r)+\func{\done-\hone}{q}(\htwo\cdot q^r+\cone+\ctwo).
\]
Now, because we have
\[
\sum_{j=0}^{t-1}\func{\done-\hone-j}{q}(q^r)+\sum_{j=0}^{\hone-\tlength}\func{d'-t'+1-j}{q}(q^r)=\sum_{j=0}^{\hone-1}\func{d'-t'+1-j}{q}(q^r)=\func{d'-t'+1}{q}(\hone\cdot q^r),
\]
\eqref{eq-diff} can be simplified further to
\[
v_d(\vec{a'})-v_d(\vec{a})=\func{\done-\hone}{q}(\htwo\cdot q^r+\cone+\ctwo)
-(\func{\done-\hone}{q}(\cone)+\func{\dtwon}{q}(\htwo\cdot q^r+\ctwo)).
\]
Notice that $\dtwon\geq \done-\hone$;
so, by Corollary \ref{prop-monotone},
$\func{\dtwon}{q}(\htwo\cdot q^r+\ctwo)\geq \func{\done-\hone}{q}(\htwo\cdot q^r+\ctwo)$.
Hence we can apply Lemma \ref{lem-subadd} by using $x=\cone$, $y=\htwo\cdot q^r+\ctwo$ to get
\[
\begin{split}
\func{\done-\hone}{q}(\cone)+\func{\dtwon}{q}(\htwo\cdot q^r+\ctwo)
&\geq \func{\done-\hone}{q}(\cone)+\func{\done-\hone}{q}(\htwo\cdot q^r+\ctwo)\\
&\geq \func{\done-\hone}{q}(\htwo\cdot q^r+\cone+\ctwo)
\end{split}
\]
which leads to
$v_d(\vec{a'})-v_d(\vec{a})\leq 0$.

\bigskip\noindent
\underline{\textsc{Case 2}}\quad $t=\tlength$: \quad
In this case,
we must have $\csum\geq t\cdot (\hone+1)\cdot q^r$,
which is equivalent to
\[
\htwo\cdot q^r+\cone+\ctwo\geq \tlength\cdot q^r
\]
So $c=\csum-t\cdot (\hone+1)\cdot q^r=\htwo\cdot q^r+\cone+\ctwo-\tlength\cdot q^r$.
In this case,
\eqref{eq-diff} becomes
\[
\begin{split}
v_d(\vec{a'})-v_d(\vec{a})=&\sum_{j=0}^{t-1}\func{\done-\hone-j}{q}(q^r)+\func{\done-t}{q}(c)-(\func{\done-\hone}{q}(\cone)+\func{\dtwon}{q}(\htwo\cdot q^r+\ctwo))\\
=&(\func{\done-\hone}{q}(t\cdot q^r)+\func{\dtwon}{q}(c))-(\func{\done-\hone}{q}(\cone)+\func{\dtwon}{q}(\htwo\cdot q^r+\ctwo)).
\end{split}
\]
We observe here that $t\cdot q^r+c=\cone+\htwo\cdot q^r+\ctwo$
and let $\ttwo=\lfloor \frac{\htwo\cdot q^r+\cone+\ctwo}{q^r}\rfloor$.
Because $c=\htwo\cdot q^r+\cone+\ctwo-t\cdot q^r$,
we have $\lfloor \frac{c}{q^r}\rfloor=\ttwo-t$.
Letting $c'=c-(\ttwo-t)\cdot q^r<q^r$,
then we have
\[
\func{\dtwon}{q}(c)=\func{\dtwon}{q}((\ttwo-t)\cdot q^r)+\func{\dtwon-(\ttwo-t)}{q}(c')
\]
Notice that $\ttwo\leq \htwo+1$ and $t\geq 1$, hence $\ttwo-t\leq \htwo$, so
\[
\func{\dtwon}{q}(\htwo\cdot q^r+\ctwo))=\func{\dtwon}{q}((\ttwo-t)\cdot q^r)+\func{\dtwon-(\ttwo-t)}{q}((\htwo-(\ttwo-t))\cdot q^r+\ctwo)).
\]
So we can plug  \eqref{eq-expand-a}  and \eqref{eq-expand-a'} into
\eqref{eq-diff} to get
\begin{align}\label{eq-diff-2}
v_d(\vec{a'})-v_d(\vec{a})
=&(\func{\done-\hone}{q}(t\cdot q^r)+\func{\dtwon-(\ttwo-t)}{q}(c'))\\
-&(\func{\done-\hone}{q}(\cone)+\func{\dtwon-(\ttwo-t)}{q}((\htwo-(\ttwo-t))\cdot q^r+\ctwo)).\nonumber
\end{align}
We now break this case into sub-cases depending on the relationship  between $\done-\hone$ and $\dtwon-\htwo$.

\bigskip\noindent
\underline{\textsc{Sub-Case} 2a}\quad   $\done-\hone< \dtwon-\htwo$:\qquad
Since $t\geq 1$,
$t\cdot q^r>\cone$.
So by Corollary \ref{prop-lifting},
we have
\[
\func{\done-\hone}{q}(t\cdot q^r)-\func{\done-\hone}{q}(\cone)\leq \func{\dtwon-\htwo}{q}(t\cdot q^r)-\func{\dtwon-\htwo}{q}(\cone)
\]
So \eqref{eq-diff-2} can be reduced as
\begin{align}\label{eq-diff-3}
v_d(\vec{a'})-v_d(\vec{a})
\leq &(\func{\dtwon-\htwo}{q}(t\cdot q^r)+\func{\dtwon-(\ttwo-t)}{q}(c'))\\
-&(\func{\dtwon-\htwo}{q}(\cone)+\func{\dtwon-(\ttwo-t)}{q}((\htwo-(\ttwo-t))\cdot q^r+\ctwo)).\nonumber
\end{align}
First suppose that $\htwo-(\ttwo-t)=0$: \qquad
Since $\htwo\leq \ttwo\leq \htwo+1$ and $t\geq 1$,
this can only happen when $\ttwo=\htwo+1$ and $t=1$.
For notational convenience,
let $d'=\dtwon-\htwo$,
then we only need to argue that
\begin{align}\label{eq-subadd-2}
\func{d'}{q}(q^r)+\func{d'}{q}(c')\leq \func{d'}{q}(\cone)+\func{d'}{q}(\ctwo).
\end{align}
where $q^r+c=\cone+\ctwo$.
To see this,
let $x=\max\{\cone,\ctwo\}$ and $y=\min\{\cone,\ctwo\}$ and note that
$c<y\leq x<q^r$.
By Corollary \ref{prop-lifting},
$\func{d'}{q}(q^r)-\func{d'}{q}(x)\leq \func{d'+1}{q}(q^r)-\func{d'+1}{q}(x)$.
Therefore
\[
\begin{split}
(\func{d'}{q}(q^r)+\func{d'}{q}(c'))-( \func{d'}{q}(\cone)+\func{d'}{q}(\ctwo))
&=(\func{d'}{q}(q^r)+\func{d'}{q}(c'))-(\func{d'}{q}(x)+\func{d'}{q}(y))\\
&\leq (\func{d'+1}{q}(q^r)+\func{d'}{q}(c'))-(\func{d'+1}{q}(x)+\func{d'}{q}(y))\\
&=\func{d'+1}{q}(q^r+c')-(\func{d'+1}{q}(x)+\func{d'}{q}(y)).
\end{split}
\]
Now,
consider the induction hypothesis
applied to $\vec{b}=(x,y,0,\cdots,0)$,
since $x+y<\|\vec{a}\|_1$,
we have $\func{d'+1}{q}(x)+\func{d'}{q}(y)=v_{d'+1}(\vec{b})\geq \func{d'+1}{q}(\|\vec{b}\|_1)=\func{d'+1}{q}(x+y)=\func{d'+1}{q}(q^r+c')$.
This implies
\[
(\func{d'}{q}(q^r)+\func{d'}{q}(c'))-( \func{d'}{q}(\cone)+\func{d'}{q}(\ctwo))\leq 0
\]
which gives the result.

\medskip\noindent
Alternatively, assume that $\htwo-(\ttwo-t)>0$:\qquad
Notice that $\ttwo\geq \htwo$,
so $t\geq \htwo-(\ttwo-t)$.
Hence $$\func{\dtwon-\htwo}{q}(t\cdot q^r)=\func{\dtwon-\htwo}{q}((\htwo-(\ttwo-t))\cdot q^r)+\func{\dtwon-\htwo-((\htwo-(\ttwo-t)))}{q}((\ttwo-\htwo)\cdot q^r).$$
Now,
$\dtwon-\htwo<\dtwon-(\ttwo-t)$ and $(\htwo-(\ttwo-t))\cdot q^r\geq q^r>c'$,
so we can apply Corollary \ref{prop-lifting} to get
\[
\begin{split}
&\func{\dtwon-\htwo}{q}(t\cdot q^r)+\func{\dtwon-(\ttwo-t)}{q}(c')\\
&=\func{\dtwon-\htwo}{q}((\htwo-(\ttwo-t))\cdot q^r)+\func{\dtwon-\htwo-((\htwo-(\ttwo-t)))}{q}((\ttwo-\htwo)\cdot q^r)+\func{\dtwon-(\ttwo-t)}{q}(c')\\
&\leq \func{\dtwon-(\ttwo-t)}{q}((\htwo-(\ttwo-t))\cdot q^r)+\func{\dtwon-\htwo-((\htwo-(\ttwo-t)))}{q}((\ttwo-\htwo)\cdot q^r)+\func{\dtwon-\htwo}{q}(c').
\end{split}
\]
Plugging into \eqref{eq-diff-3},
we get
\[
\begin{split}
&v_d(\vec{a'})-v_d(\vec{a})\\
&\quad\leq (\func{\dtwon-(\ttwo-t)}{q}((\htwo-(\ttwo-t))\cdot q^r)+\func{\dtwon-\htwo-((\htwo-(\ttwo-t)))}{q}((\ttwo-\htwo)\cdot q^r)+\func{\dtwon-\htwo}{q}(c'))\\
&\qquad-(\func{\dtwon-\htwo}{q}(\cone)+\func{\dtwon-(\ttwo-t)}{q}((\htwo-(\ttwo-t))\cdot q^r+\ctwo))\\
&\quad=(\func{\dtwon-\htwo-((\htwo-(\ttwo-t)))}{q}((\ttwo-\htwo)\cdot q^r)+\func{\dtwon-\htwo}{q}(c'))\\
&\qquad -(\func{\dtwon-\htwo}{q}(\cone)+\func{\dtwon-\htwo}{q}(\ctwo))\\
&\quad\leq (\func{\dtwon-\htwo}{q}((\ttwo-\htwo)\cdot q^r)+\func{\dtwon-\htwo}{q}(c'))\\
&\qquad-(\func{\dtwon-\htwo}{q}(\cone)+\func{\dtwon-\htwo}{q}(\ctwo)).\\
\end{split}
\]
If $\ttwo=\htwo$,
then we can apply Lemma~\ref{lem-subadd} to show that $v_d(\vec{a'})-v_d(\vec{a})\leq 0$.
Otherwise, if $\ttwo=\htwo+1$,
then we are back to inequality in the form of \eqref{eq-subadd-2},
which also gives us $v_d(\vec{a'})-v_d(\vec{a})\leq 0$ as required.

\bigskip\noindent

\underline{\textsc{Case} 2b} \quad   $\done-\hone\geq \dtwon-\htwo$:\qquad
Notice that $(\done-\hone)-(\dtwon-\htwo)=t+\htwo-\hone$,
and $\ttwo\leq \htwo+1\leq \hone$,
so
\[
\begin{split}
&\func{\dtwon-(\ttwo-t)}{q}((\htwo-(\ttwo-t))\cdot q^r+\ctwo)\\
&=\func{\dtwon-(\ttwo-t)}{q}((\hone-\ttwo)\cdot q^r)+\func{\done-\hone}{q}((t+\htwo-\hone)\cdot q^r)
+\func{\dtwon-\htwo}{q}(\ctwo)
\end{split}
\]
and
\[
\func{\done-\hone}{q}(t\cdot q^r)=\func{\done-\hone}{q}((t+\htwo-\hone)\cdot q^r)+\func{\dtwon-\htwo}{q}((\hone-\htwo)\cdot q^r).
\]
Therefore
\eqref{eq-diff-2} can be rewritten as
\[
\begin{split}
v_d(\vec{a'})-v_d(\vec{a})
=&(\func{\dtwon-\htwo}{q}((\hone-\htwo)\cdot q^r)+\func{\dtwon-(\ttwo-t)}{q}(c'))\\
-&(\func{\done-\hone}{q}(\cone)+\func{\dtwon-(\ttwo-t)}{q}((\hone-\ttwo)\cdot q^r)
+\func{\dtwon-\htwo}{q}(\ctwo)).\nonumber
\end{split}
\]
Notice that $\dtwon-(\ttwo-t)-(\hone-\ttwo)=\done-\hone$;
so, by the fact that
$(\hone-\ttwo)\cdot q^r+\cone<q^{r+1}$,
we actually have $$\func{\done-\hone}{q}(\cone)+\func{\dtwon-(\ttwo-t)}{q}((\hone-\ttwo)\cdot q^r)=\func{\dtwon-(\ttwo-t)}{q}((\hone-\ttwo)\cdot q^r+\cone).$$
So, we have
\begin{align}\label{eq-diff-4}
v_d(\vec{a'})-v_d(\vec{a})
=&(\func{\dtwon-\htwo}{q}((\hone-\htwo)\cdot q^r)+\func{\dtwon-(\ttwo-t)}{q}(c'))\\
-&(\func{\dtwon-(\ttwo-t)}{q}((\hone-\ttwo)\cdot q^r+\cone)
+\func{\dtwon-\htwo}{q}(\ctwo))\nonumber
\end{align}
This is a somewhat familiar expression.
Actually,
the right-hand side of \eqref{eq-diff-4} is of the same form as \eqref{eq-diff-3}.
With a similar argument
we can also conclude that $v_d(\vec{a'})-v_d(\vec{a})\leq 0$.
\end{proof}

\newpage
\bibliographystyle{plain}
\bibliography{mybib}

\end{document}